\documentclass[runningheads]{llncs}

\newif\iflncs
\newif\ifanon
\newif\ifappendix
\lncsfalse
\anonfalse
\appendixtrue

\usepackage[T1]{fontenc}
\usepackage{amsmath,amsfonts}
\usepackage{mathtools}
\usepackage{enumitem,soul,framed,url}
\usepackage{bbm,bm,xspace,setspace}
\usepackage{xcolor,colortbl,multicol,graphicx,ifthen}
\usepackage[
    bookmarks=false,
    colorlinks=true,
    citecolor=orange,
    linkcolor=blue!55!black,
    urlcolor=blue!55!black,
    backref=page
]{hyperref}
\usepackage{braket,comment,booktabs,tabularx,array}
\usepackage{cleveref}
\usepackage{orcidlink}
\iflncs
    \usepackage{amssymb}
    \usepackage{mathrsfs}
    \usepackage[expansion=false]{microtype}
\else
    \usepackage[a4paper,margin=1in]{geometry}
    \usepackage{newpxtext,newpxmath}
    \usepackage{cite}
\fi

\newcommand{\mca}{\mathcal{A}}

\newcommand{\mci}{\mathcal{I}}
\newcommand{\mco}{\mathcal{O}}

\newcommand{\mcs}{\mathcal{S}}

\newcommand{\mcx}{\mathcal{X}}
\newcommand{\mcy}{\mathcal{Y}}

\newcommand{\msfp}{\mathsf{P}}
\newcommand{\msfq}{\mathsf{Q}}
\newcommand{\msfu}{\mathsf{U}}

\newcommand{\bits}[1][]{\ifthenelse{\equal{#1}{}}{\{0,1\}}{\{0,1\}^{#1}}}
\newcommand{\blocks}[1][]{\bits[n]}
\newcommand{\integers}[1][]{\ifthenelse{\equal{#1}{}}{\mathbb{Z}}{\mathbb{Z}/{#1}\mathbb{Z}}}
\newcommand{\field}[1][2]{\mathbb{F}_{#1}}
\newcommand{\vecsp}[2][2]{\field[#1]^{#2}}
\newcommand{\binsp}[1][n]{\vecsp{#1}}

\newcommand{\gl}[2][2]{\mathrm{GL}(#2,#1)}
\newcommand{\glbinsp}[1][n]{\gl{#1}}
\newcommand{\reals}[1][]{\ifthenelse{\equal{#1}{}}{\mathbb{R}}{\mathbb{R}_{#1}}}
\newcommand{\nonnegreals}{\reals[\geq 0]}
\newcommand{\complexes}[1][]{\ifthenelse{\equal{#1}{}}{\mathbb{C}}{\mathbb{C}_{#1}}}
\newcommand{\indicate}[1][\mci]{\mathsf{1}_{#1}}
\newcommand{\findicate}[1][\mci]{\hat{\mathsf{1}}_{#1}}
\newcommand{\trivchar}[1][]{\ifthenelse{\equal{#1}{}}{\hat{\mathbf{1}}}{\hat{\mathbf{1}}^{(#1)}}}
\newcommand{\dual}[1][G]{\widehat{#1}}
\newcommand{\supp}{\operatorname{supp}}
\newcommand{\unimes}[1][]{\ifthenelse{\equal{#1}{}}{\msfu}{\msfu_{#1}}}
\newcommand{\unidens}[1][\mci]{\mu_{#1}}
\newcommand{\funidens}[1][\mci]{\widehat{\mu_{#1}}}
\newcommand{\rp}{\pi}
\newcommand{\rf}{\mathsf{g}}
\newcommand{\sop}[1][k]{\rp^{#1}_{+}}

\newcommand{\compdist}[3]{\ensuremath{\left\| {#2} - {#3} \right\|_{#1}}}
\newcommand{\pr}[2][]{\ifthenelse{\equal{#1}{}}{\Pr\left({#2}\right)}{\Pr_{#1}\left({#2}\right)}}
\newcommand{\floor}[1]{\lfloor{#1}\rfloor}
\newcommand{\ceil}[1]{\lceil{#1}\rceil}

\ifanon
\newcommand{\AJ}[1]{}
\newcommand{\RB}[1]{}
\newcommand{\XG}[1]{}
\newcommand{\MergeNote}[1]{}
\else
\newcommand{\AJ}[1]{{\color{green!50!black}{\textbf{[Ashwin: #1]}}}}
\newcommand{\RB}[1]{{\color{magenta}{\textbf{[Ritam: #1]}}}}
\newcommand{\XG}[1]{{\color{cyan!50!black}{\textbf{[Xiaoning: #1]}}}}
\newcommand{\MergeNote}[1]{{\color{blue!55!black}{\textbf{[Merge: #1]}}}}
\fi
\makeatletter
\let\NAT@parse\undefined
\makeatother
\renewcommand*{\backref}[1]{}
\renewcommand*{\backrefalt}[4]{%
    \ifcase #1 \textsc{(not cited.)}%
    \or \textsc{\textit{(cited on p.~#2.)}}%
    \else \textsc{\textit{(cited on pp.~#2.)}}%
    \fi
}

\ifanon
\author{}
\institute{}
\else
\author{%
    Ritam Bhaumik\inst{1}\orcidlink{0000-0002-2883-4870}%
    \and%
    Chun Guo\inst{2,3,4}\orcidlink{0000-0002-8520-6301}%
    \and%
    Xiaoning Guo\inst{2,3,4}\orcidlink{0009-0004-6091-9630}%
    \and%
    Ashwin Jha\inst{5}\orcidlink{0000-0001-5957-2837}%
}
\institute{
    CRC, TII, Abu Dhabi, UAE\\
    \email{\textcolor{magenta!90!black}{bhaumik.ritam@gmail.com}}
    \and
    School of Cyber Science and Technology;\\
    \and
    State Key Laboratory of Cryptography and Digital Economy Security;\\
    \and
    Key Laboratory of Cryptologic Technology and Information Security of Ministry of Education,\\
    Shandong University, Qingdao, China\\
    \email{\textcolor{magenta!90!black}{chun.guo.sc@gmail.com,xxiaoningguo@gmail.com}}
    \and
    University of Wuppertal, Wuppertal, Germany\\
    \email{\textcolor{magenta!90!black}{ashwin.jha@outlook.de}}
}
\authorrunning{R. Bhaumik, C. Guo, X. Guo, A. Jha}
\fi

\title{Indistinguishability of Sum of Permutations\\
{\normalsize{A Fourier Analytic Route to Classical and Quantum Security}}}
\titlerunning{Indistinguishability of Sum of Permutations}

\begin{document}
\maketitle

\begin{abstract}
We study classical and quantum indistinguishability of sums of independent random permutations and related transformations from permutations to functions. Let $G$ be a finite abelian group of order $N$, and let $\sop[k](x)=\rp_1(x)+\cdots+\rp_k(x)$ for $k\geq2$ independent uniform random permutations of $G$. We give a unified Fourier analytic treatment in which the construction is represented by its probability density and a distinguisher by its acceptance function, with the classical and quantum query models imposing different restrictions on the Fourier support of the latter.

Classically, we obtain the bound $O_k(q/N^{k-1/2})$ for every $q<N$, and refine it below the birthday threshold to $O_k(q^2/N^k)$. In the quantum model, a simulation argument gives $O_k(N^{-(k-3/2)})$ for $q\leq(N-1)/2$, while Fourier interpolation gives concrete finite bounds up to $q\leq4N/15$ and the query-dependent bounds
\[
    O\!\left(\min\left\{N^{-1/2},\frac{q^3}{N^2}+\frac1N\right\}\right),
    \qquad
    O_k\!\left(\min\left\{\frac{q^3}{N^k},N^{-(k-3/2)}\right\}\right),
\]
for $k=2$ and $k \geq 3$, respectively, throughout $1\leq q\leq(N-1)/2$. For $q = 1$, the first bound sharpens to $O(N^{-2})$. Over $G=\mathbb F_2^n$, a one-query Fourier attack matches the order of our one-query bound, while an $N/2$-query parity attack with advantage $1/2$ shows that our bounds reach the constant-advantage query threshold.

We further study two variants of sum of permutations over binary vector spaces. First, we allow arbitrary surjective linear postprocessing, which includes truncation, and obtain classical and quantum bounds that retain the output-size dependence. Second, we analyse Dinur's variable-output single-permutation construction, $\mathsf{LXoP}$, for every fixed output width, and derive its classical and quantum security bounds; for one- and two-block outputs, we give concrete quantum security bounds.

\end{abstract}

\section{Introduction}
Block ciphers, or pseudorandom permutation (PRP) candidates, are among the most thoroughly studied primitives in symmetric-key cryptography, while many applications~\cite{RogawayS2006:SIV:EUROCRYPT,Chen2022:KDF:NIST-SP,Nandi2024:HMAC-NMAC:INDOCRYPT} require a pseudorandom function (PRF). Apart from a few examples like \textsf{SURF} \cite{Bernstein1997:SURF:ONLINE}, \textsf{SipHash} \cite{AumassonB2012:SIPHASH:INDOCRYPT}, and some \textsf{AES}-based constructions~\cite{MenninkN2017:AES-PRF:TOSC,BanikILMS2021:ORTHROS:TOSC,FlorezGutierrezGLST2024:ZIP-AES:ASIACRYPT}, relatively few dedicated PRF candidates have been proposed or analysed. This makes permutation-to-function transformation a natural route to pseudorandom function designs. It is well-known~\cite{Rao1966:WOR:RISI,Freedman1977:WOR:JASA,BellareKR1994:CBC-MAC:CRYPTO,ChangN2008:PRP-PRF:EPRINT} that a random permutation itself can be distinguished from a random function by looking for collisions, with constant advantage after about $\surd N$ classical queries, where $N$ denotes the ambient group size. A natural question is therefore how to combine a small number of random permutation calls so that the resulting function remains indistinguishable from random well beyond the birthday bound.

In this line of research, three simple constructions are the sum of permutations~\cite{BellareKR1998:PRP-to-PRF:EUROCRYPT}, encrypted Davies-Meyer~\cite{CogliatiS2016:EDM:CRYPTO} and its dual~\cite{MenninkN2017:EDMD:CRYPTO}. In particular, the sum of permutations construction has received a lot of attention from the community. Let $G$ be a finite abelian group of order $N$, written additively, and let $\rp_1,\ldots,\rp_k$ be independent uniform random permutations of $G$. For $k\geq2$, define the sum of $k$ permutations by
\[
    \sop[k](x)\coloneq\rp_1(x)+\cdots+\rp_k(x).
\]
When $G=\binsp$, this is the usual XOR-of-permutations construction. The two-permutation case goes back to Bellare et al.~\cite{BellareKR1998:PRP-to-PRF:EUROCRYPT}, and the $k$-permutation generalisation~\cite{Lucks2000:SOP:EUROCRYPT} is due to Lucks. The construction is simple to evaluate: the permutation calls are independent and can be made in parallel, followed by a group addition. Its security analysis, however, has proved considerably more involved, and has occupied the community for more than two decades.

\paragraph{Classical security of sum of permutations.}
A long line of work~\cite{BellareI1999:TOOL:EPRINT,Lucks2000:SOP:EUROCRYPT,Patarin2008:SOP:ICITS,CogliatiLP2014:XORP:FSE,MenninkP2015:XORP:ACNS,Eberhard2017:SOP:arXiv,DaiHT2017:Chi-Squared:CRYPTO,DuttaNS2022:Mirror-Theory:IEEE-IT,CogliatiP2020:MT:EPRINT,Dinur2024:SOP:EUROCRYPT} has investigated how much security this simple transformation provides. Adding independent permutation outputs allows collisions, so the absence of collisions no longer distinguishes the construction from a random function. The outputs nevertheless remain dependent, because each constituent permutation samples without replacement. An analysis must bound the effect of these dependencies on the joint distribution of all queried outputs, including when the number of queries is close to $N$. Over the years, different proof techniques have been used to prove this result when $G=\binsp$, with various degrees of success. The first verifiable proof~\cite{Lucks2000:SOP:EUROCRYPT} was given by Lucks who employed the game-playing technique to prove PRF advantage of $O(q^3/N^{2})$, where $q$ is the number of queries. Bellare and Impagliazzo, and later, Dai et al. devised new statistical techniques~\cite{BellareI1999:TOOL:EPRINT,DaiHT2017:Chi-Squared:CRYPTO} to show bounds of $O_n(q/N)$ and $O(\sqrt{q^{3}/N^{3}})$, respectively. Patarin's mirror theory~\cite{Patarin2008:SOP:ICITS,Patarin2010:MT:EPRINT,Patarin2016:MT:EPRINT} has been the main tool to study the problem from a combinatorial viewpoint, and a verified proof~\cite{DuttaNS2022:Mirror-Theory:IEEE-IT} of the mirror theory result yields a bound of $O(q^2/N^2)$ for $q \leq N/17$ (see also~\cite{CogliatiP2020:MT:EPRINT} for the single-permutation variant). On the other hand, for $\surd{N} \leq q < N$, the best-known attack~\cite{Patarin2013:Attack-SOP:ACNS} due to Patarin only achieved an advantage of $\Omega(\sqrt{q^2/N^3})$, leaving a tightness gap. This gap was recently filled~\cite{Dinur2024:SOP:EUROCRYPT} by Dinur who established a bound of $O(\sqrt{q^2/N^3})$ for $q < N/2$ and $N \geq 1000$, which is tight for $q \geq \surd N$. Shortly after, it came to light that Eberhard had already proved~\cite{Eberhard2017:SOP:arXiv} the same bound for any finite abelian group and for every $q < N$, in a work~\cite{Eberhard2017:SOP:arXiv} that was largely unknown in the cryptography community. Both Eberhard and Dinur applied Fourier analytic techniques to study the problem. Dinur also studied the sum of $k$ permutations construction, and obtained a bound of $q/N^{k-1/2}$ for $q < N/2$ queries.

\paragraph{Variants of sum of permutations.} Several variants of the sum of permutations have been proposed to obtain shorter outputs or to generate several output blocks from a single permutation. For shorter outputs, summation can be combined with the \textit{truncation} technique~\cite{HallWKS1998:PRP-to-PRF:CRYPTO}. Let $G=\binsp$, with $|G|=N$, and let $\tau\colon G\to H$ be a surjective linear map, where $|H|=M\leq N$. The corresponding construction is
\[
    x\longmapsto\tau\bigl(\pi_1(x)+\pi_2(x)\bigr),
\]
where $\pi_1$ and $\pi_2$ are independent random permutations of $G$. When $\tau$ is a coordinate projection, Choi et al.~\cite{ChoiKLL2022:TSOP:ASIACRYPT} proved a classical distinguishing bound of $O(q\sqrt{qM}/N^2)$ for $q\leq N/4$.

For outputs in the larger range $G^w$, Iwata's $\mathsf{XORP}[w]$ construction~\cite{Iwata2006:CENC:FSE} uses a single random permutation $\pi\colon G\to G$ and is given by
\[
    x\longmapsto
    \bigl(\pi(0\parallel x)+\pi(1\parallel x),\ldots,
          \pi(0\parallel x)+\pi(w\parallel x)\bigr).
\]
Here the prefixes are encoded in $s=\lceil\log_2(w+1)\rceil$ bits and $x\in\binsp[n-s]$. The construction produces $w$ output blocks using $w+1$ permutation calls. For fixed $w$, its classical distinguishing advantage is $O_w(q/N)$~\cite{Patarin2016:MT:EPRINT,IwataMV2016:CENC:EPRINT,BhattacharyaN2018:VOL-XORP:TOSC}. Dinur~\cite{Dinur2025:More-SOP:EUROCRYPT} introduced the $\mathsf{LXoP}$ family, which also produces $w$ output blocks from $w+1$ calls, but combines consecutive permutation outputs using a linear automorphism $\sigma\colon G\to G$:
\[
    x\longmapsto
    \bigl(\pi(0\parallel x)+\sigma(\pi(1\parallel x)),\ldots,
          \pi((w-1)\parallel x)+\sigma(\pi(w\parallel x))\bigr).
\]
For $w\in\{1,2\}$, he proved a classical distinguishing advantage bound of $O(q/N^{3/2})$ for $q \leq N/16$ and $q \leq N/32$, respectively, and $N \geq 1024$.

\paragraph{Quantum security of sum of permutations.}
The quantum setting raises a different difficulty. For these input-symmetric constructions, a classical distinguisher can be reduced to observing the outputs on a fixed set of queried inputs. This remains true even when the queries are adaptive, as the labels of fresh inputs do not affect the distribution of the next answer. A quantum distinguisher may instead query a superposition of inputs and combine information from different positions by interference. Its behaviour is therefore not determined by the output distribution on any one fixed set of $q$ inputs.

Quantum security of XOR-of-PRP constructions has been studied~\cite{MenninkS2017:XORP-Q:PQCRYPTO} by Mennink and Szepieniec, including key-recovery attacks and a security analysis. There, the construction oracle is queried classically, while the adversary may use offline quantum computation. This is the so-called Q1 security model. In this paper, we study information-theoretic security in the Q2 model, where the adversary can make quantum superposition queries to the construction. The adversary is computationally unbounded, and we count only queries to the construction. Thus the question concerns what can be learned from the oracle itself, independently of the time complexity of processing its answers.

The sum of permutations has a global constraint that makes this question particularly interesting. Every complete truth table satisfies
\[
    \sum_{x\in G}\sop[k](x)=k\sum_{y\in G}y.
\]
For a uniform random function, the sum on the left is uniform on $G$, so the same equality holds with probability exactly $1/N$. The complete-table distribution of $\sop[k]$ therefore has statistical distance at least $1-1/N$ from uniform. At the same time, the classical bounds show that every proper restriction is close to uniform. How much of this dependence can a quantum algorithm detect with few queries? At what query count does the checksum become visible?

Our bounds give a sharp answer to the second question in the binary setting. For every fixed $k\geq2$, the advantage tends to zero for all $q\leq N/2-1$, whereas an $N/2$-query parity algorithm has advantage $1/2$. For $q < N/2$, we exploit the well-known fact~\cite{BealsBCMW1998:POLY-METHOD:FOCS,Zhandry2012:2q-wise:CRYPTO} that the acceptance probability of a $q$-query quantum adversary depends only on the joint distributions of oracle outputs on all sets of at most $2q$ inputs. This allows us to focus on dependencies among small collections of outputs, even though each quantum query may involve a superposition over the entire domain.

\subsection{Our Results}
We develop a common analytic framework for classical and quantum distinguishing games, and apply it to sums of independent permutations, linear-output variants, and a variable-output construction based on a single permutation. In what follows, all bounds are information-theoretic, and the number $k$ of independent permutations, or the output width $w$, is fixed in the corresponding asymptotic statements.

\paragraph{Sum of independent permutations.}
For every finite abelian group $G$ of order $N$, we prove the classical bounds
\[
    \compdist{(q)}{\sop[k]}{\rf}
    \leq
    \begin{cases}
        0, & 0\leq q\leq1,\\
        O_k(q^2/N^k), & 2\leq q\leq\sqrt N,\\
        O_k(q/N^{k-1/2}), & \sqrt N<q<N,
    \end{cases}
\]
where $\rf\colon G\to G$ is a uniform random function. This gives a single proof over arbitrary finite abelian groups, and refines the bound when $q < \surd N$. For $k = 2$, the bound above that threshold is Eberhard's bound; our proof extends it to every fixed $k \geq 2$.

For quantum queries, we obtain
\[
\begin{aligned}
    \compdist{(q)}{\ket{\sop[2]}}{\ket\rf}
    &\leq O\!\left(\min\left\{
        N^{-1/2},\frac{q^3}{N^2}+\frac1N
    \right\}\right),\\
    \compdist{(q)}{\ket{\sop[k]}}{\ket\rf}
    &\leq O_k\!\left(\min\left\{
        \frac{q^3}{N^k},N^{-(k-3/2)}
    \right\}\right)\qquad(k\geq3),
\end{aligned}
\]
throughout $1\leq q\leq(N-1)/2$. For $k=2$ and $q=1$, the bound improves to $O(N^{-2})$. The bounds involving $q^3$ describe the behaviour in the range $q < \surd N$, while the bounds independent of $q$ remain small throughout the stated range. We also give concrete bounds for $N \geq 1024$ and $q \leq 4N/15$. 

Over $G=\binsp$, a one-query algorithm has advantage $1/(2(N-1)^k)$, matching the order of our upper bound at one query. For $n\geq2$, the $N/2$-query parity attack, together with the uniform bound of $O_k(N^{-(k-3/2)})$ throughout $q \leq N/2-1$, establishes the exact number of queries required to achieve a constant advantage.

\paragraph{Linear-output-postprocessing.}
For $G=\binsp$, let $\tau\colon G\to H=\binsp[m]$ be surjective and linear, and put $M=2^m$. We prove a classical bound $O_k(q\sqrt M/N^k)$ for $\tau\sop[k]$ throughout $q < N$. In particular, for the truncated construction considered~\cite{ChoiKLL2022:TSOP:ASIACRYPT} by Choi et al., our bound improves their single-user bound $O(q\sqrt{qM}/N^2)$ by a factor $\sqrt q$. We also obtain quantum bounds that retain the same dependence on $M$, both for small query counts and throughout $q\leq(N-1)/2$.

\paragraph{Variable output from a single permutation.}
We analyse Dinur's $\mathsf{LXoP}$ family for every fixed width $w$, under the admissibility condition stated in Section~\ref{sec:lxop}. The construction produces $w$ output blocks using $w+1$ permutation calls. We obtain classical advantage $O_w(q/N^{3/2})$ and quantum advantage $O_w(N^{-1/2})$ for query counts up to a sufficiently small constant fraction of the construction's input-domain size. We also give concrete quantum bounds for $w=1,2$.

\subsection{Technical Overview and Organisation}
\paragraph{Densities and query restrictions.}
In Section~\ref{sec:framework}, we represent the construction by its probability density $\varphi$ and a distinguisher by its acceptance probability function $a$. Relative to the reference product measure, the distinguishing advantage is then given by the inner product $|\langle a,\varphi-1\rangle|$. To analyse this inner product, we first choose an orthonormal basis for functions on the output set, with respect to the reference output distribution, containing the constant function $\indicate[]$. We use this same basis for each truth-table entry $f(x)$. Taking products of these basis functions, with one factor for each input, gives an orthonormal basis for functions of the complete truth-table. We expand $a$ and $\varphi$ in this product basis and group the terms according to the set $S$ of inputs at which the factor is nonconstant. These groups give the orthogonal \textit{pure-support} components. A component indexed by $S$ depends only on the outputs at inputs in $S$, and averaging over any one of these outputs, with the others fixed, gives zero. Orthogonality then expresses the distinguishing advantage as the absolute value of a sum of inner products between $a$- and $\varphi$-components with matching supports.

This decomposition is defined for an arbitrary product measure, and the grouped components do not depend on the particular basis chosen. Under the uniform measure on a finite abelian group, we may choose the group characters as our basis. The components are then precisely the terms of the Fourier expansion grouped by their support. For the constructions studied in this paper, input symmetry lets us restrict a classical distinguisher to a fixed set $Q$ of $q$ inputs. Only supports contained in $Q$ can contribute to its advantage. For a quantum distinguisher, Lemma~\ref{lem:fourier-zhandry-method}, shows that only components supported on at most $2q$ inputs can contribute, although different components may involve different inputs. Since $0\leq a\leq1$, Cauchy-Schwarz bounds the remaining inner product by one half of the $L_2$ norm of the corresponding density projection. This yields the Fourier interpolation bound of Theorem~\ref{thm:quantum-fourier}.

\paragraph{Injection moments and the first nonzero level.}
Section~\ref{sec:sop} applies this framework to sums of independent permutations. On a fixed set of distinct inputs, the outputs of one permutation form a uniformly random injection. Summing independent permutation outputs convolves their densities, so the Fourier coefficients of the sum are powers of those of one injection. Eberhard's fourth-moment bound~\cite{Eberhard2017:SOP:arXiv}, together with a bound on the largest nonconstant coefficient, gives the global classical bound for every $k\geq2$. We also derive a higher-moment bound for the Fourier coefficients of the injection density, which controls the components involving three or more inputs. In the classical analysis, we apply it to the joint output distribution on $q$ fixed inputs. In the quantum analysis, we apply it to the corresponding components of the distribution of the complete truth table, retaining only those involving at most $2q$ inputs. The detailed coefficient bounds and moment calculations are given in Appendices~\ref{sec:explicit-max-moment} and~\ref{sec:extended-injection-proof}.

We call the collection of components supported on exactly $t$ inputs \textit{level $t$}. Since the one-point marginals are uniform, level one vanishes. To obtain a more precise bound at small query counts, we treat level two separately before applying Cauchy-Schwarz. This component is a scalar multiple of the difference between the number of output collisions and its expected value under a uniform random function. The classical proof bounds its contribution in $L_1$ by summing over the queried pairs. For quantum queries, Appendix~\ref{sec:sop-query-dependent} bounds the correlation of this collision count with the acceptance function. We obtain an $O(q^3)$ bound using Zhandry's small-range method~\cite{Zhandry2012:SRD:FOCS}: the centred collision count is the first-order deviation from a uniform function when outputs are sampled through a large but finite intermediate range. The common moment estimate bounds the remaining levels.

\paragraph{Simulation from the checksum.}
To reach every quantum query count below $N/2$, in Section~\ref{sec:simulation}, we use Eberhard's bound along with a checksum-based simulation argument. Any $N-1$ outputs of $\sop[k]$ determine the last output. Start instead with independent uniform values on $N-1$ inputs and complete the table using the same checksum. The resulting function is uniform among all functions with that checksum and is $(N-1)$-wise independent. It is therefore indistinguishable from a uniform random function whenever $2q\leq N-1$. It remains to compare the sum of permutations with this completed reference function. Both complete tables are determined by their restrictions to the chosen $N-1$ inputs, so their statistical distance is bounded by the classical estimate on those inputs. This gives the quantum bound $O_k(N^{-(k-3/2)})$ throughout the required range.

\paragraph{Linear-output-postprocessing and shared permutation calls.}
Section~\ref{sec:linear-output} considers applying a surjective linear map $\tau$ to each output. A Fourier character of the projected output $\tau(y)$ can also be viewed as a character of the original output $y$: its index $\eta$ is replaced by $\tau^\top\eta$. Thus only indices in the subspace $\operatorname{Im}(\tau^\top)$, which has size $M$, enter the analysis. To bound their contribution, we use the fact that applying the same invertible linear map to every output of a random injection preserves its distribution. This symmetry lets us average over subspaces of size $M$ and compare the Fourier moments restricted to such a subspace with the unrestricted moments. This is where the output size $M$ enters the bounds. For the refined quantum estimate, we also distinguish tuples of character indices according to the dimension of their span. The probability that all indices lie in a random subspace of size $M$ depends on this dimension. Keeping this dependence gives a sharper bound for the higher levels.

For the $\mathsf{LXoP}$ family studied in Section~\ref{sec:lxop}, each construction output consists of $w$ blocks obtained from $w+1$ permutation calls. Expanding a character of the construction output gives a product of characters of the underlying permutation outputs, with one index for each call. These permutation outputs form a random injection, so we can again use its Fourier coefficients. The shared calls force the indices associated with each construction input to satisfy a linear relation. We exploit this relation when applying a Dinur-style recursive identity~\cite{Dinur2025:More-SOP:EUROCRYPT} that replaces an injection coefficient by a sum of coefficients involving fewer permutation outputs. To obtain a bound on the sum of squared coefficients, we must control how many original index tuples can give the same tuple after a reduction. By choosing the reductions in a suitable order, the linear relation and the admissibility condition let us recover the removed index. This controls the number of times each reduced coefficient is counted and gives the required second-moment bound. Summing over the supports permitted in each query model then gives the classical and quantum bounds.
\section{Fourier Analytic Framework}
\label{sec:framework}
This section first fixes the notation used throughout the paper and then develops a unified analytic framework for classical and quantum distinguishing games. The framework is inspired by the Fourier analytic methods employed by Eberhard~\cite{Eberhard2017:SOP:arXiv} and Dinur~\cite{Dinur2024:SOP:EUROCRYPT,Dinur2025:More-SOP:EUROCRYPT}.

\paragraph{Notation.}
For a positive integer $n$, write $[n] \coloneq \{1,\ldots,n\}$ and $[0] \coloneq \varnothing$. For $0\leq k\leq n$, let $(n)_k\coloneq n(n-1)\cdots(n-k+1)$, with $(n)_0=1$. Empty sums and products are understood as zero and one, respectively. The constants in $O_k(\cdot)$ and $O_w(\cdot)$ may depend only on the indicated fixed parameter.

Let $\msfq$ be a probability measure on a non-empty finite set $\Omega$. We write $x\sim\msfq$ to mean that $x$ is drawn according to $\msfq$, and write $\supp(\msfq)$ for its support. We denote the uniform measure on $\Omega$ by $\unimes[\Omega]$, abbreviated to $\unimes$ when $\Omega$ is clear. A \textit{density} relative to $\msfq$ is a function $\varphi \colon \supp(\msfq) \to \nonnegreals$ satisfying $\mathbf E_{x\sim\msfq}[\varphi(x)]=1$. We sometimes identify $\varphi$ with the probability distribution that assigns mass $\varphi(z)\msfq(z)$ to each $z\in\supp(\msfq)$, when the reference measure is clear. The reference measure $\msfq$ itself has density $\indicate[]$.

Let $L_2(\msfq)$ be the space of complex-valued functions on $\supp(\msfq)$. For $h,g\in L_2(\msfq)$ and $1\leq p<\infty$, write
\[
    \langle h,g\rangle_{\msfq}
    \coloneq
    \mathbf E_{x\sim\msfq}
    [h(x)\overline{g(x)}],
    \qquad
    \|h\|_{p,\msfq}
    \coloneq
    \bigl(\mathbf E_{x\sim\msfq}[|h(x)|^p]\bigr)^{1/p}.
\]
For densities $\varphi$ and $\psi$ relative to $\msfq$, their statistical distance is given by
\[
    \mathrm{TVD}_\msfq(\varphi,\psi) \coloneq \frac12\|\varphi-\psi\|_{1,\msfq},
\]
whence by Cauchy-Schwarz,
\begin{equation}\label{eq:density-norm-tower}
    \mathrm{TVD}_\msfq(\varphi,\psi)
    \leq
    \frac12\|\varphi-\psi\|_{2,\msfq}.
\end{equation}

\subsection{Product Spaces and Support Decomposition}
\label{subsec:abstract-product-space}
We recall the standard support decomposition on a finite product space (see~\cite[Sections~8.1 and~8.3]{ODonnell2014:AOBF:Book}), which requires no algebraic structure on the output set.

Let $\mcx$ and $\mcy$ be finite sets, with $\mcy\neq\varnothing$. Fix a probability measure $\msfp$ on $\mcy$ and equip $\mcy^{\mcx}$ with the product measure $\msfp^{\mcx}$. Unless stated otherwise, densities, expectations, inner products, and norms on this space are taken with respect to $\msfp^{\mcx}$, whose subscript is omitted.

Let
\(
    L_2^0(\msfp)
    \coloneq
    \left\{
        g\in L_2(\msfp):
        \mathbf E_{y \sim \msfp}[g(y)]=0
    \right\}.
\)
Choose an orthonormal basis $B_0$ of $L_2^0(\msfp)$ and put $B=\{\indicate[]\}\cup B_0$. For $\lambda=(\lambda_x)_{x\in\mcx}\in B^{\mcx}$, define
\[
    \Phi_\lambda(f)
    \coloneq
    \prod_{x\in\mcx}\lambda_x(f(x)),
    \qquad
    \supp(\lambda)
    \coloneq
    \{x\in\mcx:\lambda_x\neq\indicate[]\}.
\]
We call $\supp(\lambda)$ the \emph{pure support} of $\Phi_\lambda$. The functions $\Phi_\lambda$ form an orthonormal basis of $L_2(\msfp^{\mcx})$.

For $S\subseteq\mcx$, let $\Pi_S$ be the orthogonal projection onto the subspace
\[
    \mathcal H_S
    \coloneq
    \operatorname{Span}
    \{\Phi_\lambda:\supp(\lambda)=S\}.
\]
Then
\begin{equation}\label{eq:abstract-support-decomposition}
    L_2(\msfp^{\mcx})
    =\bigoplus_{S\subseteq\mcx}\mathcal H_S,
    \qquad
    h=\sum_{S\subseteq\mcx}\Pi_Sh,
    \qquad
    \|h\|_2^2
    =\sum_{S\subseteq\mcx}\|\Pi_Sh\|_2^2.
\end{equation}
The spaces $\mathcal H_S$ depend only on the orthogonal decomposition $L_2(\msfp) = \operatorname{Span}\{1\}\oplus L_2^0(\msfp)$, and not on the particular basis chosen for $L_2^0(\msfp)$. The space $\mathcal H_{\varnothing}$ consists of the constant functions, so $\Pi_{\varnothing}h=\mathbf E[h]\cdot\indicate[]$. In particular, $\Pi_{\varnothing}\varphi=\indicate[]$ for every density $\varphi$. If the distribution represented by $\varphi$ has marginal $\msfp$ at $x$, then
\begin{equation}\label{eq:vanishing-first-level}
    \Pi_{\{x\}}\varphi
    =\mathbf E[\varphi\mid f(x)]-\indicate[]=0.
\end{equation}
For integers $t,d\geq0$, write
\[
    h_{=t}\coloneq\sum_{|S|=t}\Pi_Sh,
    \qquad
    \Pi_{\leq d}h
    \coloneq
    \sum_{\substack{S\subseteq\mcx\\|S|\leq d}}\Pi_Sh.
\]
We shall repeatedly use the following coordinate dependence property of the product basis. A function $h\in L_2(\msfp^{\mcx})$ depends only
on the coordinates in $Q\subseteq\mcx$ if and only if
\begin{equation}\label{eq:coordinate-dependence-support}
    \Pi_Sh=0
    \qquad\text{for every }S\not\subseteq Q.
\end{equation}
If these projections vanish, the expansion of $h$ contains only basis functions supported in $Q$, which depend only on $Q$. Conversely, if $h$ depends only on $Q$, every coefficient involving a mean zero basis function at a coordinate outside $Q$ vanishes after averaging over that coordinate.

\subsection{Distinguishing Advantage and Query Support Cutoffs}
\label{subsec:support-distinguishing}
For any function $f \in \mcy^\mcx$, we denote the classical oracle corresponding to $f$ by the function itself, while the quantum oracle is denoted by $\ket{f}$.

Let $\mco_0$ and $\mco_1$ be oracles sampled according to densities $\varphi$ and $\psi$ relative to $\msfp^{\mcx}$, respectively, and accessed in the same model. For an oracle algorithm $\mca$ and a fixed function $f$, let $a(f)\in[0,1]$ be the probability that $\mca$ outputs $1$ given oracle access to $f$. Define
\[
    \compdist{\mca}{\mco_0}{\mco_1}
    \coloneq
    \left|
        \Pr[\mca^{\mco_0}=1]
        -
        \Pr[\mca^{\mco_1}=1]
    \right|.
\]
We write $\compdist{(q)}{\mco_0}{\mco_1}$ for the supremum over all algorithms making at most $q$ queries in the indicated access model. We work in the information-theoretic setting: algorithms may be computationally unbounded, and only oracle queries are counted. Then,
\begin{equation} \label{eq:oracle-to-density} 
    \compdist{\mca}{\mco_0}{\mco_1}
    =
    \left|\mathbf E[a\varphi]-\mathbf E[a\psi]\right|
    =
    \left|
        \langle a,\varphi-\psi\rangle
    \right|
    \leq
    \mathrm{TVD}_{\msfp^{\mcx}}(\varphi,\psi),
\end{equation}
where the final inequality follows by writing $\langle a,\varphi-\psi\rangle = \mathbf E[(a-\tfrac12)(\varphi-\psi)]$, since $\mathbf E[\varphi-\psi]=0$ and $|a-\tfrac12|\leq\tfrac12$, so that the inner product is at most $\frac{1}{2}\|\varphi-\psi\|_1$. For the remainder of this paper, we assume $\psi = \indicate[]$.

We next apply the support decomposition~\eqref{eq:abstract-support-decomposition} to the inner product in~\eqref{eq:oracle-to-density}. The constant component of $\varphi-1$ is zero, and its projection onto every non-empty support $S$ is $\Pi_S\varphi$. Orthogonality therefore gives
\begin{equation}\label{eq:supportwise-advantage}
    \compdist{\mca}{\mco_0}{\mco_1}
    =\left|\sum_{\varnothing\neq S\subseteq\mcx}
        \langle\Pi_Sa,\Pi_S\varphi\rangle\right|.
\end{equation}
For $j\geq1$, put
\[
    \Delta_j(a,\varphi)
    \coloneq
    \sum_{\substack{S\subseteq\mcx\\|S|=j}}
        \langle\Pi_Sa,\Pi_S\varphi\rangle.
\]
The first two levels will often be treated directly, while the remaining levels are bounded in $L_2$.

\begin{lemma}\label{lem:layered-fourier-bound}
    Let $\mco_0 \sim \varphi$ and $\mco_1 \sim \indicate[]$ be function oracles in the same access model. Suppose $\Pi_Sa = 0$ whenever $|S| > d$. Then, for every integer $0 \leq r \leq d$,
    \[
        \compdist{\mca}{\mco_0}{\mco_1}
        \leq
        \sum_{j=1}^{r} |\Delta_j(a,\varphi)|
        +\frac12
        \left(
            \sum_{\substack{S\subseteq\mcx\\r+1\leq|S|\leq d}}
                \|\Pi_S\varphi\|_2^2
        \right)^{1/2}.
    \]
    If $a$ depends only on the coordinates in $Q\subseteq\mcx$, the sum in the remainder may be restricted to $S\subseteq Q$.
\end{lemma}
\begin{proof}
    By~\eqref{eq:supportwise-advantage}, only non-empty supports contribute to the advantage. Since $\Pi_Sa=0$ for $|S|>d$, we have
    \[
        \langle a,\varphi-1\rangle
        =
        \sum_{j=1}^{r}\Delta_j(a,\varphi)
        +
        \sum_{\substack{S\subseteq\mcx\\r<|S|\leq d}}
            \langle\Pi_Sa,\Pi_S\varphi\rangle.
    \]
    Put $m=\mathbf E[a]$. Since $0\leq a\leq1$,
    \[
        \sum_{\varnothing\neq S\subseteq\mcx}\|\Pi_Sa\|_2^2
        =\|a-m\|_2^2
        =\mathbf E[a^2]-m^2
        \leq m-m^2
        \leq\frac14.
    \]
    By Cauchy-Schwarz,
    \[
        \left|\sum_{\substack{S\subseteq\mcx\\r<|S|\leq d}}
            \langle\Pi_Sa,\Pi_S\varphi\rangle\right|
        \leq
        \frac12
        \left(
            \sum_{\substack{S\subseteq\mcx\\r<|S|\leq d}}
                \|\Pi_S\varphi\|_2^2
        \right)^{1/2}.
    \]
    The triangle inequality gives the claim. If $a$ depends only on $Q$, then~\eqref{eq:coordinate-dependence-support} gives $\Pi_Sa=0$ whenever $S\not\subseteq Q$, so these supports may also be omitted.\qed
\end{proof}

\begin{remark}\label{rem:classical-localisation}
    The constructions considered below are invariant in distribution under input permutations, as is the product reference measure. Thus, for $q\leq|\mcx|$, a classical $q$-query distinguisher may be restricted to any fixed set $Q$ of size $q$: answer its $j$-th fresh query using the $j$-th point of $Q$, and replay answers to repeated queries. Input symmetry preserves the transcript distribution under both oracle distributions. By~\eqref{eq:coordinate-dependence-support}, the resulting acceptance function satisfies $\Pi_Sa_Q=0$ whenever $S\not\subseteq Q$.
\end{remark}

For a $q$-query quantum distinguisher, every nonzero pure-support component of the acceptance probability involves at most $2q$ input coordinates. We use the standard quantum function-oracle model, with a fixed injective encoding of $\mcy$ into a finite abelian group and translation by the encoded output; all applications below instantiate it by the usual addition oracle on a finite abelian output group.

\begin{lemma}\label{lem:fourier-zhandry-method}
Let $\mca$ be a quantum algorithm making at most $q$ queries to a function $f\colon\mcx\to\mcy$, and let $a(f)$ be its acceptance probability. Then
\[
    \Pi_Sa=0
    \qquad\text{whenever }|S|>2q.
\]
\end{lemma}
\begin{proof}
    Write the function oracle as
    \[
        U_f\ket{x,r,z}
        =
        \ket{x}\otimes V_{f(x)}\ket r\otimes\ket z,
    \]
    where $V_y\ket r=\ket{r+\operatorname{enc}(y)}$ translates the answer register by the fixed encoding of $y\in\mcy$, and $z$ denotes the internal registers of the algorithm. For a basis label $u=(x,r,z)$, every matrix entry $\langle u|U_f|u'\rangle$ therefore depends on $f$ only through the coordinate $f(x)$.
    
    Let $\rho_t(f)$ denote the density matrix of the algorithm after the $t$-th query and the subsequent oracle-independent operation, with $\rho_0$ the initial state. We prove inductively that every matrix entry of $\rho_t(f)$, viewed as a function of $f$, is a linear combination of functions each depending on at most $2t$ input coordinates.
    
    The claim is immediate for $t=0$, since $\rho_0$ is independent of $f$. Suppose it holds after $t-1$ queries, and consider the state immediately after the next oracle call,
    \[
        \rho'_t(f)=U_f\rho_{t-1}(f)U_f^\dagger.
    \]
    For basis labels $u=(x,r,z)$ and $v=(x',r',z')$, we have
    \[
    \begin{aligned}
        \langle u|\rho'_t(f)|v\rangle
        =\sum_{u',v'}
            &\langle u|U_f|u'\rangle\,
            \overline{\langle v|U_f|v'\rangle}\\
            &{}\cdot\langle u'|\rho_{t-1}(f)|v'\rangle.
    \end{aligned}
    \]
    The first two factors depend only on $f(x)$ and $f(x')$. By the induction hypothesis, the last factor is a linear combination of functions each depending on a set $Q\subseteq\mcx$ with $|Q|\leq2(t-1)$. Multiplying any such function by the first two factors produces a function depending only on
    \[
        Q\cup\{x,x'\},
    \]
    which has size at most $2t$. Hence every entry of $\rho'_t(f)$ has the required form.
    
    If $\Phi_t$ denotes the oracle-independent operation following the query, then every entry of $\rho_t(f)=\Phi_t(\rho'_t(f))$ is a fixed linear combination of entries of $\rho'_t(f)$, with coefficients independent of $f$. We emphasise that the density-matrix formulation also allows mixed initial states, classical randomness, and intermediate measurements, with their outcomes retained in the internal registers. The same coordinate-dependence bound therefore holds for $\rho_t(f)$. This completes the induction.
    
    Finally, for a fixed acceptance measurement operator $M$,
    \[
        a(f)=\operatorname{Tr}(M\rho_q(f))
    \]
    is a fixed linear combination of entries of $\rho_q(f)$. Thus $a$ is a linear combination of functions each depending on at most $2q$ input coordinates. For any such function $h$, depending only on $Q$, and any $S$ with $|S|>2q$, we have $S\not\subseteq Q$. By~\eqref{eq:coordinate-dependence-support}, $\Pi_Sh=0$. Linearity of $\Pi_S$ now gives
    \[
        \Pi_Sa=0
        \qquad\text{whenever }|S|>2q.\tag*{\qed}
    \]
\end{proof}
Lemmas~\ref{lem:layered-fourier-bound} and~\ref{lem:fourier-zhandry-method}, with $r=0$ and $d=2q$, yield the following Fourier interpolation result, which we use repeatedly.
\begin{theorem}[Fourier interpolation]\label{thm:quantum-fourier}
    Let $\mco_0 \sim \varphi$ and $\mco_1 \sim \indicate[]$ be quantum function oracles. For every $q\geq0$,
    \[
        \compdist{(q)}{\mco_0}{\mco_1}
        \leq
        \frac12\|\Pi_{\leq2q}(\varphi-1)\|_2
        =
        \frac12
        \left(
            \sum_{\substack{S\subseteq\mcx\\1\leq|S|\leq2q}}
                \|\Pi_S\varphi\|_2^2
        \right)^{1/2}.
    \]
\end{theorem}

\subsection{Character Bases over Finite Abelian Groups}\label{subsec:fourier-analysis}
We now instantiate the support decomposition of Subsection~\ref{subsec:abstract-product-space}. Let $G$ be a finite abelian group of order $N\geq2$, written additively, and set $\mcy=G$ and $\msfp=\unimes[G]$. Unless stated otherwise, all densities, expectations, inner products, and norms below are relative to the corresponding uniform measures.

Let $\dual$ denote the character group of $G$. A character $\psi\colon G\to\complexes$ satisfies $\psi(x+y)=\psi(x)\psi(y)$ and $|\psi(x)|=1$, with $\psi^{-1}(x)=\overline{\psi(x)}$. The principal character $\trivchar$ is the constant function $\indicate[]$. The characters form a canonical orthonormal basis of $L_2(\unimes[G])$, abbreviated to $L_2(G)$; the nonprincipal characters form a basis of its mean zero subspace. Thus, in Section~\ref{subsec:abstract-product-space}, we may take $B_0=\dual\setminus\{\trivchar\}$ and $B=\dual$. We refer to~\cite{Terras1999:Fourier:Book,Babai2002:Fourier:Online} for the standard background in Fourier analysis.

For $\chi=(\chi_x)_{x\in\mcx}\in\dual^{\mcx}$, the corresponding product basis function from Section~\ref{subsec:abstract-product-space} is
\[
    \chi(f)
    \coloneq
    \Phi_\chi(f)
    =
    \prod_{x\in\mcx}\chi_x(f(x)),
    \quad
    \supp(\chi)
    =
    \{x\in\mcx:\chi_x\neq\trivchar\},
    \quad
    |\chi|
    \coloneq
    |\supp(\chi)|.
\]
For $h\colon G^{\mcx}\to\complexes$, define its Fourier coefficient at $\chi$ by
\[
    \widehat h(\chi)
    \coloneq
    \langle h,\chi\rangle
    =
    \mathbf E_{f\sim\unimes[G^{\mcx}]}
    [h(f)\overline{\chi(f)}].
\]
Fourier inversion and Parseval's identity give
\begin{equation}\label{eq:fourier-inversion}
    h
    =
    \sum_{\chi\in\dual^{\mcx}}\widehat h(\chi)\chi,
    \qquad
    \|h\|_2^2
    =
    \sum_{\chi\in\dual^{\mcx}}|\widehat h(\chi)|^2.
\end{equation}
The support projections from Section~\ref{subsec:abstract-product-space} become
\[
    \Pi_Sh
    =
    \sum_{\supp(\chi)=S}\widehat h(\chi)\chi,
    \qquad
    h_{=t}
    =
    \sum_{|\chi|=t}\widehat h(\chi)\chi.
\]
Thus $h_{=t}$ is the level $t$ Fourier component. We call $\{(\psi,\ldots,\psi):\psi\in\widehat G\}$ the diagonal subgroup of $\widehat G^{\mcx}$ and write $\trivchar[\mcx]$ for the principal character on $G^{\mcx}$. For $\ell>0$,
\[
    \sum_{|\chi|=t}|\widehat h(\chi)|^\ell
\]
is the $\ell$-th Fourier moment of $h$ at level $t$. For functions $h_1,h_2$ on $G^{\mcx}$, define their convolution by
\[
    (h_1*h_2)(f)
    \coloneq
    \mathbf E_{g\sim\unimes[G^{\mcx}]}
    [h_1(g)h_2(f-g)].
\]
Then
\[
    \widehat{h_1*h_2}(\chi)
    =
    \widehat h_1(\chi)\widehat h_2(\chi),
    \qquad
    \widehat{h^{*k}}(\chi)
    =
    \widehat h(\chi)^k.
\]

\paragraph{Walsh masks.}
\label{subsec:binary-walsh}
When $G=\mathbb F_2^n$, we identify $G$ with its character group through
\[
    \chi_w(x)=(-1)^{\langle w,x\rangle},
    \qquad
    \chi_u\chi_v=\chi_{u+v}.
\]
Here $\langle\cdot,\cdot\rangle$ is the binary dot product. We call $w$ a \emph{mask} and use the same term for tuples of masks. The principal character corresponds to $0$. For a linear map $L$ between binary vector spaces,
\[
    \chi_\eta(Lx)=\chi_{L^\top\eta}(x),
\]
so pulling a character back through $L$ corresponds to applying $L^\top$ to its mask.

\subsection{Injection Densities and Fourier Coefficients}\label{subsec:probability-density}
Let $\Omega$ be a finite ambient set and let $\varnothing\neq A\subseteq\Omega$. Write $\indicate[A]$ for its indicator and
\[
    \unidens[A]
    =\frac{|\Omega|}{|A|}\indicate[A]
\]
for the density of $\unimes[A]$ relative to the ambient measure $\unimes[\Omega]$.

\subsubsection{Uniform Density on Injections.}
For $0\leq d\leq N$, let $\mci_d$ denote the set of injections $[d]\to G$. Its density relative to $\unimes[G^d]$ is given by
\[
    \unidens[\mci_d]=\frac{N^d}{(N)_d}\indicate[\mci_d].
\]
We also refer to this as the injection density over $G^d$. Fixing an ordering of $G$ identifies $\mci_N$ with the set $\mcs$ of permutations of $G$, and $G^N$ with the complete function space $G^G$. In particular,
\[
    \unidens[\mcs]=\frac{N^N}{N!}\indicate[\mcs].
\]
Restricting a uniform element of $\mci_d$ to any $t$-subset $S \subseteq [d]$ gives a uniform element of $\mci_t$. Indeed, each injective assignment of $t$ distinct outputs to $S$ has $(N-t)_{d-t}$ extensions to an injection on $[d]$, and $(N)_d = (N)_t(N-t)_{d-t}$. Accordingly, if $\chi\in\widehat{G}^d$ has support $S=\{i_1,\ldots,i_t\}$, with $i_1<\cdots<i_t$, then
\begin{equation}\label{eq:injection-restriction}
    \funidens[\mci_d](\chi)
    =\funidens[\mci_t](\chi_{i_1},\ldots,\chi_{i_t}).
\end{equation}
Stated differently, inserting principal characters does not change the corresponding Fourier coefficient of the injection density. We also use~\eqref{eq:injection-restriction} when $d=N$, with $\unidens[\mci_N]$ identified with $\unidens[\mcs]$.

Two further symmetries will be useful. Negating every output preserves the injection distribution, so its Fourier coefficients are real. Translating every output by a fixed $z\in G$ also preserves this distribution. Consequently,
\[
    \funidens[\mci_d](\chi)=0
    \quad\text{unless}\quad
    \prod_{i=1}^d\chi_i=\trivchar.
\]
Indeed, translation multiplies the coefficient by $\overline{\prod_i\chi_i(z)}$, which must be one for every $z$ if the coefficient is nonzero.

Since permuting the input coordinates preserves the uniform injection distribution, the fixed-support moments of the injection density depend only on the size of the support. For $1\leq t\leq N$ and $\ell>0$, we write
\[
    \Theta_{t,\ell}
    \coloneqq
    \sum_{\alpha\in(\widehat G\setminus\{\trivchar\})^t}
    \left|
        \funidens[\mci_t](\alpha)
    \right|^\ell,
    \qquad
    \Lambda_t
    \coloneqq
    \max_{\alpha\in(\widehat G\setminus\{\trivchar\})^t}
    \left|
        \funidens[\mci_t](\alpha)
    \right|,
\]
and abbreviate $\Theta_t\coloneqq\Theta_{t,2}$. Thus, $\Theta_{t,\ell}$ is the $\ell$-th Fourier moment of $\unidens[\mci_d]$ on one fixed support of size $t$, while $\Lambda_t$ is the corresponding maximum Fourier coefficient. Accordingly, for every $1\leq t\leq d\leq N$,
\[
    \sum_{\substack{\chi\in\widehat G^d\\|\chi|=t}}
    \left|
        \funidens[\mci_d](\chi)
    \right|^\ell
    =
    \binom{d}{t}\Theta_{t,\ell},
    \qquad\text{and}\qquad
    \sum_{\substack{\chi\in\widehat G^d\\
                    \chi\neq\trivchar[d]}}
    \left|
        \funidens[\mci_d](\chi)
    \right|^\ell
    =
    \sum_{t=1}^d \binom{d}{t}\Theta_{t,\ell},
\]
where the first quantity is the $\ell$-th Fourier moment
of $\unidens[\mci_d]$ at level $t$. The quantities $\Theta_t$ and $\binom{d}{t}\Theta_{t}$ are referred to as the Fourier weights for $\unidens[\mci_t]$ and $\unidens[\mci_d]$ at level $t$ in~\cite{Dinur2024:SOP:EUROCRYPT}, respectively. For $k\geq2$, bounding all but two factors by the maximum gives
\[
    \Theta_{t,2k}
        \leq
        \Lambda_t^{2k-2}\Theta_t.
\]
We shall also use Eberhard's coefficient-merging identity~\cite[Section~4]{Eberhard2017:SOP:arXiv}. For $2\leq t\leq N$ and nonprincipal characters $\chi_1,\ldots,\chi_t\in\widehat G$,
\begin{equation}\label{eq:injection-coefficient-merging}
    \funidens[\mci_t](\chi_1,\ldots,\chi_t)
    =
    -\frac{1}{N-t+1}
    \sum_{i=1}^{t-1}
    \funidens[\mci_{t-1}]
    (\chi_1,\ldots,\chi_i\chi_t,\ldots,\chi_{t-1}).
\end{equation}
Condition on the first $t-1$ distinct outputs $y_1,\ldots,y_{t-1}$. Since $\chi_t$ is nonprincipal,
\[
    \mathbf E\left[\overline{\chi_t(y_t)} \mid y_1,\ldots,y_{t-1}\right]
    =
    -\frac{1}{N-t+1}\sum_{i=1}^{t-1}\overline{\chi_t(y_i)}.
\]
Multiplying by $\prod_{j=1}^{t-1}\overline{\chi_j(y_j)}$ and taking expectation over the uniform injection on the first $t-1$ coordinates gives the identity.
\section{Sum of Independent Permutations}\label{sec:sop}
Given $k\geq2$ independent uniform random permutations $\rp_1,\ldots,\rp_k$ of $G$, define
\[
    \sop[k](x)\coloneq\rp_1(x)+\cdots+\rp_k(x),
    \qquad x\in G.
\]
For $G=\mathbb F_2^n$, this is the usual $k$-XOR-of-permutations construction. We use $k$ for the number of permutations and reserve $t$ for a Fourier level. The density on any fixed set of $d<N$ distinct inputs is
\[
    \varphi_{d,k}=\unidens[\mci_d]^{*k},
\]
so its Fourier coefficients are $\funidens[\mci_d](\chi)^k$. We first record the two global estimates that will be used throughout the section, starting with Eberhard's result for $k=2$.
\begin{proposition}[Theorem~1.5 in~\cite{Eberhard2017:SOP:arXiv}] \label{prop:eberhard-fourth-moment}
     Let $G$ be a finite abelian group of order $N$. For every $1\leq d<N$,
    \[
        \sum_{t=1}^d\binom dt\Theta_{t,4}
        =\sum_{\chi\in\dual^d\setminus\{\trivchar[d]\}}
            |\funidens[\mci_d](\chi)|^4
        =\|\unidens[\mci_d]^{*2}-1\|_2^2
        \leq O\left(\frac{d^2}{N^3}\right).
    \]
\end{proposition}
\begin{proof}
    The inequality follows directly from~\cite[Theorem~1.5]{Eberhard2017:SOP:arXiv}. The two equalities follow from Parseval, convolution-multiplication duality, and a grouping of the characters by their Fourier level.\qed
\end{proof}

We next record a slight generalisation of the smoothing argument used by Dinur that isolates the only additional step needed for $k>2$.
\begin{proposition} \label{prop:sop-fourier-smoothing}
    For every $1\leq d<N$, every $k\geq2$, and every set $\Gamma\subseteq\widehat G^d\setminus\{\trivchar[d]\}$,
    \begin{equation}\label{eq:sop-fourier-smoothing}
        \sum_{\chi\in\Gamma}|\funidens[\mci_d](\chi)|^{2k}
        \leq
        \lambda_N^{2k-4}
        \sum_{\chi\in\Gamma}|\funidens[\mci_d](\chi)|^4,
    \end{equation}
    where $\lambda_N\coloneq\binom N2^{-1/2}$. In particular,
    \begin{equation}\label{eq:sop-global-moment}
        \sum_{t=1}^d\binom dt\Theta_{t,2k}
        \leq O_k\left(\frac{d^2}{N^{2k-1}}\right).
    \end{equation}
\end{proposition}
\begin{proof}
    Fix $\chi\neq\trivchar[d]$, and pad it with principal characters to a character $\chi'\in\widehat G^N$. By~\eqref{eq:injection-restriction},
    \[
        \funidens[\mci_d](\chi)=\funidens[\mcs](\chi').
    \]
    Since $d<N$ and $\chi$ is nonprincipal, the padded tuple has both a principal and a nonprincipal coordinate, and is not diagonal. Eberhard's non-diagonal coefficient bound~\cite[Proof of Theorem~1.4, p.~18]{Eberhard2017:SOP:arXiv}, after multiplication by $N^N/N!$ to pass to densities, gives $|\funidens[\mci_d](\chi)|\leq\lambda_N$. Hence
    \[
        |\funidens[\mci_d](\chi)|^{2k}
        \leq
        \lambda_N^{2k-4}|\funidens[\mci_d](\chi)|^4.
    \]
    Summing over $\Gamma$ proves~\eqref{eq:sop-fourier-smoothing}. Taking all nonprincipal characters and applying Proposition~\ref{prop:eberhard-fourth-moment} gives~\eqref{eq:sop-global-moment}.\qed
\end{proof}

\begin{remark} \label{rem:complete-permutation-density-applicable}
    The same bound also applies to the complete permutation density $\unidens[\mcs]$ for every character $\chi\in\dual^G$ with $1\leq |\chi|<N$. Indeed, if $S=\supp(\chi)$ and $|S|=t<N$, then the restriction of a uniform permutation to $S$ is a uniform injection, and hence $\widehat{\unidens[\mcs]}(\chi) = \widehat{\unidens[\mci_t]}(\chi|_S)$.
\end{remark}

The next estimate controls the levels above two in both query models. The parameter $m$ counts the input positions over which these levels are summed: we use $m = q$ for the classical case and $m = N$ for the quantum case.

\begin{lemma}\label{lem:extended-injection-moments}
    Let $G$ be a finite abelian group of order $N\geq1024$. For every integer $3\leq T\leq8N/15$,
    \begin{equation}\label{eq:extended-injection-fourth-moment}
        \sum_{t=3}^{T}\binom Nt\Theta_{t,4}<\frac{13}{4N^2},
    \end{equation}
    and
    \begin{equation}\label{eq:extended-injection-maximum}
        \max_{3\leq t\leq T}\Lambda_t\leq\binom N4^{-1/2}.
    \end{equation}
    Consequently, for integers $k\geq2$ and $T\leq m\leq N$,
    \begin{equation}\label{eq:common-injection-tail}
        \sum_{t=3}^{T}\binom mt\Theta_{t,2k}
        \leq\frac{\binom m3}{\binom N3}
            \frac{13}{4N^2}\binom N4^{-(k-2)}.
    \end{equation}
\end{lemma}
The proof is given in Appendix~\ref{sec:extended-injection-proof}.

\subsection{Refined Classical Security}\label{subsec:sop-classical}
The $O_k(q/N^{k-1/2})$ bound for every $q < N$ follows from the preceding smoothing argument. To improve the bound below the birthday threshold, we first compute the two-point distribution and then apply the moment bound to the levels above two.

\begin{lemma}\label{lem:sop-two-point-law}
    Let $x\neq x'$ and let $F=\sop[k]$. The density $\varphi_{2,k}$ of $(F(x),F(x'))$ on $G^2$ satisfies
    \[
        \varphi_{2,k}(y,y')-1
        =\eta_k\bigl(N\indicate[\{y\}](y')-1\bigr),
        \qquad
        \eta_k=\left(-\frac1{N-1}\right)^k.
    \]
    If $\tau\colon G\to H$ is a surjective homomorphism and $|H|=M$, the two-point density $\nu_{2,k}$ of $\tau F$ satisfies
    \begin{equation}\label{eq:sop-two-point-pushforward}
        \nu_{2,k}(y,y')-1=\eta_k\bigl(M\indicate[\{y\}](y')-1\bigr).
    \end{equation}
\end{lemma}
\begin{proof}
    For each permutation, put $Z_i=\rp_i(x')-\rp_i(x)$. The variable $\rp_i(x)$ is uniform on $G$, $Z_i$ is uniform on $G\setminus\{0\}$, and the two are independent. The pairs are independent across $i$. Hence $F(x)$ is uniform and independent of $S_k=Z_1+\cdots+Z_k$, while $F(x')=F(x)+S_k$. If $p_j(z)=\Pr[S_j=z]$, then convolution with the uniform distribution on $G\setminus\{0\}$ gives
    \[
        p_{j+1}(z)=\frac{1-p_j(z)}{N-1}.
    \]
    Subtracting $1/N$ and iterating from $p_1(0)=0$, $p_1(z)=1/(N-1)$ for $z\neq0$, yields
    \[
        p_k(z)=\frac1N+\eta_k\left(\indicate[\{0\}](z)-\frac1N\right).
    \]
    The claimed density identity follows from
    $\varphi_{2,k}(y,y')=N^2\Pr[F(x)=y,F(x')=y']=Np_k(y'-y)$. Pushing forward through $\tau$ sends uniform measure on $G$ to uniform measure on $H$, yielding~\eqref{eq:sop-two-point-pushforward}.\qed
\end{proof}

For a fixed pair $\{x,x'\}$, the corresponding support-two density projection therefore has
\[
    \|\Pi_{\{x,x'\}}\varphi\|_2^2
    =(N-1)|\eta_k|^2,
\]
and its total-variation contribution is
\[
    |\eta_k|\left(1-\frac1N\right)
    =\frac1{N(N-1)^{k-1}}.
\]
Consequently, after restricting to a fixed set $Q$ of $q$ inputs (by Remark~\ref{rem:classical-localisation}),
\begin{equation}\label{eq:sop-classical-level-two}
    |\Delta_2|
    \leq
    \min\left\{
        \frac{\binom q2}{N(N-1)^{k-1}},
        \frac{|\eta_k|}{2}\sqrt{\binom q2(N-1)}
    \right\}.
\end{equation}
The first estimate uses the pairwise total variation and the triangle inequality: since $\Pi_{\{x,x'\}}\varphi$ has mean zero, $|\langle\Pi_{\{x,x'\}}a,\Pi_{\{x,x'\}}\varphi\rangle| = |\langle a-\tfrac{1}{2},\Pi_{\{x,x'\}}\varphi\rangle| \leq \tfrac{1}{2}\|\Pi_{\{x,x'\}}\varphi\|_1$, as in~\eqref{eq:oracle-to-density}. The second is Cauchy-Schwarz on level two.

\begin{theorem}\label{thm:sop-k-classical}
    Fix $k\geq2$, and let $\rf\colon G\to G$ be uniform. For every $q<N$,
    \[
        \compdist{(q)}{\sop[k]}{\rf}
        \leq
        \begin{cases}
            0, & 0\leq q\leq1,\\[1mm]
            O_k(q^2/N^k), & 2\leq q\leq\sqrt N,\\[1mm]
            O_k(q/N^{k-1/2}), & \sqrt N<q<N.
        \end{cases}
    \]
    In particular, $\compdist{(q)}{\sop[k]}{\rf}\leq O_k(q/N^{k-1/2})$ throughout $q < N$.
\end{theorem}
\begin{proof}
    The case $q=0$ is immediate. For the global estimate with $q\geq1$, Remark~\ref{rem:classical-localisation} reduces the analysis to a fixed set of $q$ inputs. The density there is $\varphi_{q,k}=\unidens[\mci_q]^{*k}$, and Proposition~\ref{prop:eberhard-fourth-moment}, followed by Proposition~\ref{prop:sop-fourier-smoothing}, gives
    \[
        \|\varphi_{q,k}-1\|_2
        \leq O_k\left(\frac q{N^{k-1/2}}\right).
    \]
    The statistical-distance bound follows from~\eqref{eq:density-norm-tower}.
    
    For $q\leq1$, the advantage is zero because every one-point marginal is uniform. Assume now $2\leq q\leq\sqrt N$. By Remark~\ref{rem:classical-localisation}, fix a set $Q$ of $q$ queried inputs. Equation~\eqref{eq:vanishing-first-level} removes level one, while~\eqref{eq:sop-classical-level-two} gives $O_k(q^2/N^k)$ at level two. For $N\geq1024$ and $q\geq3$, apply~\eqref{eq:common-injection-tail} with $m=T=q$. Lemma~\ref{lem:layered-fourier-bound} then bounds the remainder by
    \[
        O_k\left(\frac{q^{3/2}}{N^{2k-3/2}}\right)
        =O_k\left(\frac{q^2}{N^k}\right).
    \]
    For $q=2$, the remainder is empty. The finitely many smaller group orders are absorbed into the implicit constant.\qed
\end{proof}

\subsection{Quantum Security}\label{subsec:sop-quantum}
We first state the bounds for the full range $2q<N$. They combine a direct analysis of the low-support Fourier components with the simulation argument of Section~\ref{sec:simulation}.

\begin{theorem}\label{thm:sop-quantum-combined-local}
    Let $G$ be a finite abelian group of order $N$, and let $\rf\colon G\to G$ be uniform. For $1\leq q\leq(N-1)/2$,
    \[
        \compdist{(q)}{\ket{\sop[2]}}{\ket\rf}
        \leq
        O\!\left(
            \min\left\{
                N^{-1/2},
                \frac{q^3}{N^2}+\frac1N
            \right\}
        \right).
    \]
    For every fixed $k\geq3$,
    \[
        \compdist{(q)}{\ket{\sop[k]}}{\ket\rf}
        \leq
        O_k\!\left(
            \min\left\{
                \frac{q^3}{N^k},
                \frac{1}{\sqrt{N^{2k-3}}}
            \right\}
        \right).
    \]
    For $q=1$, the $k=2$ bound sharpens to $O(N^{-2})$.
\end{theorem}
The bounds independent of $q$ follow from Corollary~\ref{cor:sop-quantum-completion}. Appendix~\ref{sec:sop-query-dependent} treats level two using Zhandry's small-range distributions and proves the query-dependent bounds, including the one-query refinement.

The direct Fourier calculation also gives explicit constants. Lemma~\ref{lem:extended-injection-moments} applies up to support size $8N/15$, which allows $q\leq4N/15$ quantum queries. This extends the range $q\leq N/4$ obtained from moment estimates restricted to support size at most $N/2$.

\begin{corollary}\label{cor:sop-quantum-concrete}
    Let $G$ be a finite abelian group of order $N\geq1024$, fix $k\geq2$, and let $\rf\colon G\to G$ be uniform. For $0\leq q\leq4N/15$, we have
    \[
        \compdist{(q)}{\ket{\sop[k]}}{\ket\rf}
        \leq
        \sqrt{
            \frac{N}{8(N-1)^{2k-2}}
            +\frac{13}{16N^2}\binom N4^{-(k-2)}
        }
        <0.355\,\frac{\sqrt N}{(N-1)^{k-1}}.
    \]
\end{corollary}
\begin{proof}
    The case $q=0$ is immediate, so assume $q\geq1$. Regrouping the complete-table Fourier coefficients by support, and using Remark~\ref{rem:complete-permutation-density-applicable} and $\Theta_1 = 0$, gives
    \[
        \sum_{1\leq|\chi|\leq2q}|\funidens[\mcs](\chi)|^{2k}
        =\sum_{t=2}^{2q}\binom Nt\Theta_{t,2k}.
    \]
    The exact support-two contribution is $N/(2(N-1)^{2k-2})$, by Lemma~\ref{lem:sop-two-point-law}. Equation~\eqref{eq:common-injection-tail}, with $m=N$ and $T=2q$, bounds the remaining contribution; this remainder is zero when $q=1$.
    Consequently,
    \[
        \sum_{1\leq|\chi|\leq2q}|\funidens[\mcs](\chi)|^{2k}
        \leq
        \frac{N}{2(N-1)^{2k-2}}
        +\frac{13}{4N^2}\binom N4^{-(k-2)}.
    \]
    Theorem~\ref{thm:quantum-fourier} gives the first inequality. For the second, $(N-1)^2\leq\binom N4$ when $N\geq1024$, so the square of the bound, after division by $N/(N-1)^{2k-2}$, is at most
    \[
        \frac18+\frac{13(N-1)^2}{16N^3}
            \left(\frac{(N-1)^2}{\binom N4}\right)^{k-2}
        \leq\frac18+\frac{13}{16\cdot1024}<0.355^2.\tag*{\qed}
    \]
\end{proof}

\subsection{Quantum Attacks over $\mathbb F_2^n$}
\label{subsec:sop-attacks}
We now take $G=\mathbb F_2^n$ and $N=2^n$. For distinct $x,x'\in G$ and nonzero $\lambda\in G$, there exists a one-query quantum algorithm that returns $b_\lambda(x,x')=\langle\lambda,f(x)\oplus f(x')\rangle$ with certainty. This is the usual two-point phase-query algorithm, written with an arbitrary output mask; see also Bonnetain et al.~\cite{BonnetainLNS2021:Linearization:ASIACRYPT}. The algorithm and its proof are given in Appendix~\ref{sec:sop-two-point-attack}.

\begin{proposition}\label{prop:sop-one-query-attack}
    For every $k\geq2$, there exists a one-query quantum algorithm that distinguishes $\ket{\sop[k]}$ from $\ket\rf$ with advantage $1/(2(N-1)^k)$.
\end{proposition}
\begin{proof}
    Fix distinct $x,x'\in G$ and nonzero $\lambda\in G$. The one-query algorithm computes $b_\lambda(x,x')$. Use the acceptance function $a(f)=\frac12(1+(-1)^{b_\lambda(x,x')})$, which accepts exactly when $b_\lambda(x,x')=0$. This bit is uniform under $\rf$. By Lemma~\ref{lem:sop-two-point-law}, $\mathbf E[(-1)^{b_\lambda(x,x')}]=(-1/(N-1))^k$ under $\sop[k]$. Hence the distinguishing advantage is $1/(2(N-1)^k)$.\qed
\end{proof}

\begin{proposition}\label{prop:sop-half-domain-attack}
    For $n\geq2$ and $k\geq2$, there exists an $N/2$-query quantum algorithm that distinguishes $\ket{\sop[k]}$ from $\ket\rf$ with advantage $1/2$.
\end{proposition}
\begin{proof}
    Fix nonzero $\lambda\in G$ and partition $G$ into $N/2$ disjoint pairs. Applying the one-query algorithm to each pair and XORing the outputs computes $\bigoplus_{x\in G}\langle\lambda,f(x)\rangle$ using $N/2$ queries. For $f=\sop[k]$, this bit is zero because every permutation ranges over $G$ and $\bigoplus_{y\in G}y=0$ for $n\geq2$. Under $\rf$, it is uniform. Accepting when it is zero gives advantage $1/2$.\qed
\end{proof}
For $G=\mathbb F_2^n$ with $n\geq2$, the one-query attack matches the asymptotic order of our quantum upper bound at $q=1$. The matching half-domain threshold is discussed in Section~\ref{sec:simulation}.

\section{Quantum Security by Completion}\label{sec:simulation}
We next use the checksum invariance property to extend quantum security to every query count below half the domain size. The argument applies more generally when any prescribed values on a fixed number of inputs have a unique completion within a given class of functions. Completing independent uniform values gives a reference function with limited independence, to which the quantum support cutoff applies.

\begin{theorem}\label{thm:completion-lift}
    Let $\mcx$ be a finite set of size $D$, let $G$ be a finite abelian group, and fix $1\leq d\leq D$. Let $\mathcal C\subseteq G^{\mcx}$ be such that, for every $S\subseteq\mcx$ of size $d$, the restriction map
    \[
        \mathcal C\longrightarrow G^S,\qquad f\longmapsto f|_S,
    \]
    is a bijection. Let $F\colon\mcx\to G$ be a random function taking values in $\mathcal C$, and let $\rf\colon\mcx\to G$ be uniform. For every $q$ with $2q\leq d$,
    \[
        \compdist{(q)}{\ket F}{\ket\rf}
        \leq\compdist{(d)}{F}{\rf}.
    \]
\end{theorem}
\begin{proof}
    Fix a set $S$ of size $d$, and let $E_S\colon G^S\to\mathcal C$ be the inverse restriction map. Then $E_S(F|_S)=F$, whereas $F_{\mathcal C}\coloneq E_S(\rf|_S)$ is uniform on $\mathcal C$. The restriction hypothesis for every $d$-point set implies that $F_{\mathcal C}$ is $d$-wise independent. Lemma~\ref{lem:fourier-zhandry-method} therefore gives
    \[
        \compdist{(q)}{\ket{F_{\mathcal C}}}{\ket\rf}=0
        \qquad(2q\leq d),
    \]
    equivalently by Zhandry's theorem~\cite[Theorem~3.1]{Zhandry2012:2q-wise:CRYPTO}.
    
    For a fixed $q$-query distinguisher $\mca$, let $a_S(T)$ be its acceptance probability when given the complete oracle $E_S(T)$. A computationally unbounded classical distinguisher can query all points of $S$ and accept with probability $a_S(T)$. Its acceptance probabilities against $F$ and $\rf$ are exactly those of $\mca$ against $F$ and $F_{\mathcal C}$. Thus
    \[
        \compdist{\mca}{\ket F}{\ket\rf}
        =\compdist{\mca}{\ket F}{\ket{F_{\mathcal C}}}
        \leq\compdist{(d)}{F}{\rf}.
    \]
    Taking the supremum over $\mca$ proves the claim.\qed
\end{proof}
The sum of permutations construction falls in the $d=D-1$ case. Indeed, we can use its checksum to reconstruct the final output.
\begin{corollary} \label{cor:sop-quantum-completion}
    Let $G$ be a finite abelian group of order $N$. Fix some $k\geq2$, and $q\leq(N-1)/2$. Let $\rf\colon G\to G$ denote a uniform random function. Then,
    \[
        \compdist{(q)}{\ket{\sop[k]}}{\ket\rf} \leq O_k\left(N^{-(k-3/2)}\right).
    \]
\end{corollary}
\begin{proof}
    Let $c_k\coloneq k\sum_{y\in G}y$ and define the set
    \[
        \mathcal C_k
        \coloneq
        \left\{f\in G^G:\sum_{x\in G}f(x)=c_k\right\}.
    \]
    Then, for any $r \in G$ and $S = G \setminus \{r\}$, the restriction map
    \[
        \mathcal C_k \longrightarrow G^S,\qquad f\longmapsto f|_S,
    \]
    is a bijection. Indeed, given any $f' \in G^S$, we get a unique $f \in \mathcal C_k$ defined as follows:
    \[
        f(x)
        \coloneq
        \begin{cases}
            f'(x), & x \neq r,\\[3pt]
            c_k-\displaystyle\sum_{x'\neq r}f'(x'), & x = r.
        \end{cases}
    \]
    Every realisation of $\sop[k]$ satisfies
    \[
        \sum_{x\in G}\sop[k](x)
        =k\sum_{y\in G}y
        =c_k,
    \]
    and hence $\sop[k] \in \mathcal C_k$. Theorem~\ref{thm:completion-lift} with $d=N-1$, followed by Theorem~\ref{thm:sop-k-classical} at $N-1$ classical queries, gives
    \[
        \compdist{(q)}{\ket{\sop[k]}}{\ket\rf}
        \leq\compdist{(N-1)}{\sop[k]}{\rf}
        \leq O_k\left(\frac{N-1}{N^{k-1/2}}\right).
    \]
    This proves the claimed bound for $2q\leq N-1$.\qed
\end{proof}

\paragraph{Sharp query threshold for the binary setting.}
When $G=\binsp$ with $n\geq2$, Corollary~\ref{cor:sop-quantum-completion} and Proposition~\ref{prop:sop-half-domain-attack} meet at consecutive query counts. For every fixed $k\geq2$, the advantage is $O_k(N^{-(k-3/2)})$ for every $q\leq N/2-1$, while an $N/2$-query algorithm has advantage at least $1/2$. Thus $N/2$ is the asymptotic constant-advantage quantum query threshold for the binary construction.

\section{Linear-Output-Postprocessing Variants}\label{sec:linear-output}
Throughout this section, $G=\mathbb F_2^n$, $N=2^n$, and we identify characters with Walsh masks. For a mask tuple $\chi\in G^{\mcx}$, write
\[
    r(\chi)\coloneq\dim\operatorname{Span}\{\chi_x:x\in\mcx\}.
\]

For $L\in\glbinsp$ and $x=(x_1,\ldots,x_d)$, write $Lx=(Lx_1,\ldots,Lx_d)$. The injection density is invariant under this diagonal action, since $L$ preserves pairwise distinctness. Averaging over the general linear group therefore controls the Fourier mass inside a fixed subspace.

\begin{proposition}\label{prop:subspace-moment-bound}
Suppose $f\colon G^d\to\reals$ is invariant under the diagonal action of $\glbinsp$, and let $\ell>0$. For every $m$-dimensional subspace $W\leq G$,
\[
    \sum_{\substack{\chi\in W^d\\\chi\neq0}}|\widehat f(\chi)|^\ell
    \leq
    \frac{|W|-1}{N-1}
    \sum_{\substack{\chi\in G^d\\\chi\neq0}}|\widehat f(\chi)|^\ell.
\]
\end{proposition}
\begin{proof}
    For $S\subseteq G^d$, write
    \[
        M_\ell(S)\coloneq\sum_{\substack{\chi\in S\\\chi\neq0}}|\widehat f(\chi)|^\ell.
    \]
    Using the invariance of $f$, we have $\widehat f(L^\top\chi)=\widehat f(\chi)$ for all $L\in\glbinsp$. It follows that $M_\ell(W^d)=M_\ell((L^\top W)^d)$. Averaging over $L\sim\unimes[\glbinsp]$ gives
    \begin{equation}\label{eq:subspace-main-bound}
        M_\ell(W^d)
        =\sum_{\substack{\chi\in G^d\\\chi\neq0}}
        |\widehat f(\chi)|^\ell
        \Pr_L\left[\operatorname{Span}\{\chi_1,\ldots,\chi_d\}\subseteq L^\top W\right].
    \end{equation}
    We claim that
    \begin{equation} \label{eq:span-containment}
        \pr{\mathrm{Span}\{\chi_1,\ldots,\chi_d\} \subseteq L^{\top}W} \leq \frac{|W|-1}{|G|-1}.
    \end{equation}
    Let $r=r(\chi)$. The subspace $L^{\top}W$ is uniform among the $m$-dimensional subspaces of $G$. For $r \leq m$, its probability of containing a fixed $r$-dimensional subspace is
    \begin{equation} \label{eq:subspace-membership-probability}
        p_{n,m}(r) = \frac{\binom{n-r}{m-r}_2}{\binom{n}{m}_2} = \prod_{i=0}^{r-1} \frac{2^m - 2^i}{2^n-2^i},
    \end{equation}
    and $p_{n,m}(r)=0$ for $r>m$. Here $\binom{\,}{\,}_2$ denotes the Gaussian binomial coefficient. Since $p_{n,m}(r)\leq p_{n,m}(1) = (|W|-1)/(N-1)$ for $r \geq 1$, claim~\eqref{eq:span-containment} follows. Substituting it in~\eqref{eq:subspace-main-bound} proves the result.\qed
\end{proof}
The same argument applies after restricting to any set of masks invariant under the diagonal action; in particular, support size and rank may be fixed simultaneously.

\subsection{Generalised and Truncated Sums of Permutations}\label{subsec:trunc-sop}
Fix $0\leq m\leq n$, let $H=\mathbb F_2^m$, $M=2^m$, and let $\tau\colon G\to H$ be a surjective linear map. Define
\[
    \tau\sop[k](x)\coloneq\tau(\sop[k](x))
    =\tau(\rp_1(x))+\cdots+\tau(\rp_k(x)).
\]
Projection onto $m$ fixed coordinates recovers the truncated sum of permutations studied in~\cite{ChoiKLL2022:TSOP:ASIACRYPT}. Proposition~\ref{prop:subspace-moment-bound} and~\eqref{eq:sop-global-moment} give the following improved classical indistinguishability bound for $\tau\sop[k]$.
\begin{theorem}\label{thm:generalised-sop-classical}
Fix $k\geq2$ and $q<N$, and let $\rf\colon G\to H$ be uniform. Then
\[
    \compdist{(q)}{\tau\sop[k]}{\rf}
    \leq O_k\left(\frac{q\sqrt M}{N^k}\right).
\]
\end{theorem}
\begin{proof}
The cases $q=0$ and $M=1$ are immediate. By Remark~\ref{rem:classical-localisation}, fix a set of $q$ distinct inputs, and let $\nu_k$ be the density of the construction on $H^q$. Character pullback gives
\begin{equation}\label{eq:generalised-sop-fourier}
    \widehat\nu_k(\eta)=\funidens[\mci_q](\tau^\top\eta)^k.
\end{equation}
Since $\tau$ is surjective, $\tau^\top$ identifies $H$ with the subspace $W=\tau^\top(H)\leq G$, of size $M$. Parseval's identity, Proposition~\ref{prop:subspace-moment-bound}, and~\eqref{eq:sop-global-moment} give
\[
    \|\nu_k-1\|_2^2
    =\sum_{\chi\in W^q\setminus\{0\}}|\funidens[\mci_q](\chi)|^{2k}
    \leq\frac{M-1}{N-1}O_k\left(\frac{q^2}{N^{2k-1}}\right)
    =O_k\left(\frac{q^2M}{N^{2k}}\right).
\]
The statistical-distance bound proves the claim.\qed
\end{proof}
For $k=2$, this improves the $O(q\sqrt{qM}/N^2)$ single-user bound of~\cite{ChoiKLL2022:TSOP:ASIACRYPT} by a factor $\sqrt q$ in their common parameter range.

\subsection{Quantum Security}
The quantum bounds retain the dependence on the output size $M$. As in the unprojected case, we combine a bound from completion with a more precise treatment of level two. Keeping the rank of the pulled-back masks improves the bound for the remaining levels.

\begin{theorem}\label{thm:generalised-sop-quantum}
    Let $\rf\colon G\to H$ be uniform. For $1\leq q\leq(N-1)/2$,
    \begin{equation}\label{eq:truncated-sop-query-dependent}
        \compdist{(q)}{\ket{\tau\sop[2]}}{\ket\rf}
        \leq O\left(\min\left\{\frac{\sqrt M}{N},\frac{q^3+M}{N^2}\right\}\right).
    \end{equation}
    For every fixed $k\geq3$,
    \[
        \compdist{(q)}{\ket{\tau\sop[k]}}{\ket\rf}
        \leq O_k\left(\min\left\{\frac{q^3}{N^k},\frac{\sqrt M}{N^{k-1}}\right\}\right).
    \]
    For $q=1$, the $k=2$ bound sharpens to $O(N^{-2})$. For $M=1$, the advantage is zero.
\end{theorem}
The proof is given in Appendix~\ref{subsec:truncated-sop-query-dependent}. We record the simulation bound and the concrete Fourier bound below.

The projected checksum determines the final output from the other $N-1$ outputs, so Theorem~\ref{thm:completion-lift} applies.
\begin{corollary} \label{cor:truncated-sop-quantum-completion}
    Let $G=\binsp$, $H=\binsp[m]$, $N=2^n$, $M=2^m$, and let $\tau\colon G\to H$ be surjective and linear. Fix $k\geq2$, and let $\rf\colon G\to H$ be uniform. For every $q\leq(N-1)/2$,
    \[
        \compdist{(q)}{\ket{\tau\sop[k]}}{\ket\rf}
        \leq O_k\left(\frac{\sqrt M}{N^{k-1}}\right).
    \]
\end{corollary}
\begin{proof}
    Every realisation $f=\tau\sop[k]$ satisfies $\sum_{x\in G}f(x)=\tau(k\sum_{y\in G}y)$. Any $N-1$ values determine the remaining value uniquely. Theorem~\ref{thm:completion-lift}, with output group $H$ and $d=N-1$, followed by Theorem~\ref{thm:generalised-sop-classical}, gives
    \[
        \compdist{(q)}{\ket{\tau\sop[k]}}{\ket\rf}
        \leq\compdist{(N-1)}{\tau\sop[k]}{\rf}
        \leq O_k\left(\frac{(N-1)\sqrt M}{N^k}\right).\tag*{\qed}
    \]
\end{proof}

For explicit constants, we restrict the Fourier mass in Corollary~\ref{cor:sop-quantum-concrete} to masks in the image of $\tau^\top$.
\begin{corollary}\label{cor:generalised-sop-quantum-concrete}
    Assume $N\geq1024$, fix $k\geq2$, and let $\rf\colon G\to H$ be uniform. For $0\leq q\leq4N/15$, we have
    \[
        \compdist{(q)}{\ket{\tau\sop[k]}}{\ket\rf}
        \leq
        \sqrt{\frac{M-1}{N-1}\left(
            \frac{N}{8(N-1)^{2k-2}}
            +\frac{13}{16N^2}\binom N4^{-(k-2)}
        \right)}.
    \]
\end{corollary}

\begin{proof}
    For the complete function table, the pullback identity~\eqref{eq:generalised-sop-fourier} holds with $\mci_q$ replaced by $\mcs$. The map $\tau^\top$ preserves support. Proposition~\ref{prop:subspace-moment-bound}, restricted to levels at most $2q$, therefore multiplies the Fourier mass used in Corollary~\ref{cor:sop-quantum-concrete} by at most $(M-1)/(N-1)$. Theorem~\ref{thm:quantum-fourier} proves the claim.\qed
\end{proof}

\section{The Variable-Output $\mathsf{LXoP}$ Family}\label{sec:lxop}
Throughout this section, $G=\mathbb F_2^n$ and $N=2^n$. We identify characters with Walsh masks. Iwata's $\mathsf{XORP}[w]$ construction~\cite{Iwata2006:CENC:FSE} produces $w$ output blocks from $w+1$ domain-separated permutation calls. Dinur defined~\cite{Dinur2025:More-SOP:EUROCRYPT} the following $\mathsf{LXoP}$ family and proposed the admissibility condition below, leaving its general-$w$ security analysis to future work. We use our notation for this family and give classical and quantum bounds for every fixed $w$.

Fix $w\geq1$, put $s=\lceil\log_2(w+1)\rceil$, and assume $s<n$. The construction has input domain $\binsp[n-s]$, of size $D=N/2^s$, and produces $w$ elements of $G$ on each input. For $j\in\{0,\ldots,w\}$, let
\[
    \iota_j\colon\binsp[n-s]\longrightarrow G,
    \qquad
    \iota_j(x)=\langle j\rangle_s\parallel x,
\]
where $\langle j\rangle_s$ denotes the $s$-bit binary representation of $j$. We call $\sigma\in\glbinsp$ \emph{$w$-admissible} if
\begin{equation}\label{eq:w-admissible}
    I+\sigma^r\in\glbinsp
    \qquad\text{for every }1\leq r\leq w.
\end{equation}
For $w=1$, this is the linear orthomorphism condition on $\sigma$. The same condition suffices for $w=2$, since $I+\sigma^2=(I+\sigma)^2$ over $\field$.

Given a uniform random permutation $\rp$ of $G$, define $\sop[\sigma,w] \colon \binsp[n-s] \to G^w$ by
\begin{equation}\label{eq:lxop-w-definition}
    \sop[\sigma,w](x)
    \coloneq
    \bigl(
        \rp(\iota_0(x))+\sigma(\rp(\iota_1(x))),\ldots,
        \rp(\iota_{w-1}(x))+\sigma(\rp(\iota_w(x)))
    \bigr).
\end{equation}
In particular, $\sop[\sigma,1]$ and $\sop[\sigma,2]$ are respectively the constructions $\mathsf{LXoP}[\sigma,n]$ and $\mathsf{LXoP}[\sigma,2,n]$ of Dinur.

The condition~\eqref{eq:w-admissible} is non-vacuous for every fixed $w$. For example, after identifying $G$ with $\field[2^n]$, multiplication by a primitive field element gives a $w$-admissible linear map whenever $w<2^n-1$.

The domain separators ensure that all $(w+1)q$ permutation inputs used by $q$ distinct construction inputs are distinct. Thus their outputs form a uniform injection, even though the $w$ output blocks at one input share permutation calls. We must account for these shared calls when bounding the Fourier moments. The resulting classical bound is as follows.

\begin{theorem}\label{thm:lxop-w-classical}
    Fix $w\geq1$. There exists a constant $c_w>0$, depending only on $w$, such that the following holds. Let $\sigma$ be $w$-admissible, let $\rf_w\colon\binsp[n-s]\to G^w$ be a uniform random function, and let $1\leq q\leq c_wD$. Then
    \[
        \compdist{(q)}{\sop[\sigma,w]}{\rf_w}
        \leq O_w\left(\frac{q}{N^{3/2}}\right).
    \]
\end{theorem}
We prove the theorem using a fixed-support moment bound that will also be used for quantum queries. The full recursive calculation is given in Appendix~\ref{sec:lxop-recursion-details}; here we describe the masks to which it applies.

Pulling an output character back through the linear combining map gives a character on the underlying permutation outputs. Write $A=\sigma^\top$ and define
\begin{equation}\label{eq:lxop-w-pullback-map}
    T_{\sigma,w}(\alpha_1,\ldots,\alpha_w)
    \coloneq
    \bigl(\alpha_1,A\alpha_1+\alpha_2,\ldots,A\alpha_{w-1}+\alpha_w,A\alpha_w\bigr)
    \in G^{w+1}.
\end{equation}
For example, for $w=2$ the output character $(\alpha_1,\alpha_2)$ pulls back to $(\alpha_1,A\alpha_1+\alpha_2,A\alpha_2)$. Extending the notation, given a mask $\alpha =(\alpha^{(1)},\ldots,\alpha^{(q)})\in(G^w)^q$, where $\alpha^{(i)} = (\alpha_{i,1},\ldots,\alpha_{i,w})\in G^w$, we apply $T_{\sigma,w}$ separately to each tuple $\alpha^{(i)}$ and use the same notation for the resulting map:
\[
    T_{\sigma,w}(\alpha)
    \coloneq
    \bigl(
        T_{\sigma,w}(\alpha^{(1)}),\ldots,
        T_{\sigma,w}(\alpha^{(q)})
    \bigr)
    \in (G^{w+1})^q.
\]
We distinguish between construction masks and injection masks. A construction mask $\alpha=(\alpha^{(1)},\ldots,\alpha^{(q)})\in(G^w)^q$ assigns one tuple of labels to each construction input. Its
\textit{construction level} is the number of inputs for which $\alpha^{(i)}\neq0$. Here a nonzero tuple may still have some zero entries. Its pullback $\beta=T_{\sigma,w}(\alpha)\in(G^{w+1})^q$ is an injection mask, with one label for each underlying permutation-output position. Its \emph{injection level} is the number of nonzero entries of $\beta$. Writing
\[
    \beta=(\beta^{(1)},\ldots,\beta^{(q)}),
    \qquad
    \beta^{(i)}
    =(\beta_{i,0},\ldots,\beta_{i,w})\in G^{w+1},
\]
we call $\beta^{(i)}$ the $i$-th \textit{block} of the injection mask.

Let $\xi^{(\sigma,w)}_{n,q}$ denote the density on any fixed set of $q$ distinct construction inputs. Expanding each output character through the linear combining map gives
\begin{equation}\label{eq:lxop-w-fourier-pullback}
    \widehat{\xi^{(\sigma,w)}_{n,q}}(\alpha)
    =\funidens[\mci_{(w+1)q}]\bigl(T_{\sigma,w}(\alpha)\bigr).
\end{equation}
Indeed, $\langle\alpha,\sigma(y)\rangle=\langle A\alpha,y\rangle$, so the character on one block has labels $(\alpha_1,A\alpha_1+\alpha_2,\ldots,A\alpha_w)$ on the corresponding permutation outputs. Those outputs form a uniform injection.

For $1\leq t\leq D$, write
\begin{equation}\label{eq:lxop-w-fixed-support-second-moment}
    L_t\coloneq
    \sum_{\substack{\alpha\in(G^w)^t\\\alpha^{(i)}\neq0\text{ for all }i}}
        |\widehat{\xi^{(\sigma,w)}_{n,t}}(\alpha)|^2
        = \sum_{\substack{\alpha\in(G^w)^t\\\alpha^{(i)}\neq0\text{ for all }i}}|\funidens[\mci_{(w+1)t}]\bigl(T_{\sigma,w}(\alpha)\bigr)|^2.
\end{equation}
Thus $L_t$ is the second Fourier moment on one fixed support of $t$ construction inputs. The corresponding injection coefficients need not all lie at the same level, because some coordinates can vanish after applying $T_{\sigma,w}$. We record the structure of the pullback before stating the moment bound. For one block, write
\[
    \beta=T_{\sigma,w}(\alpha)=(\beta_0,\ldots,\beta_w),
    \qquad C_j=A^{w-j}.
\]
By~\eqref{eq:lxop-w-pullback-map}
\begin{equation}\label{eq:lxop-w-parity-check}
    \sum_{j=0}^w C_j\beta_j
    =\sum_{j=0}^w A^{w-j}\beta_j=0.
\end{equation}
Indeed, after substitution, each $\alpha_j$ occurs twice and cancels. Conversely, the original labels are recovered by $\alpha_1=\beta_0$ and $\alpha_{j+1}=\beta_j+A\alpha_j$ for $1\leq j<w$. Hence $T_{\sigma,w}$ is injective. Since every $C_j$ is invertible, a nonzero construction block $\alpha^{(i)}$ cannot have a pullback $\beta^{(i)}$ with exactly one nonzero entry: otherwise~\eqref{eq:lxop-w-parity-check} would force that entry to be zero. Hence each nonzero construction block contributes between $2$ and $w+1$ nonzero entries to the injection mask. Consequently, construction level $t$ pulls back to injection levels between $2t$ and $(w+1)t$.

To bound $L_t$, we partition the pullbacks by their exact injection support and iterate~\eqref{eq:injection-coefficient-merging}. At each step we choose a coordinate in an untouched block. The block relation and $w$-admissibility let us recover the removed label, so distinct masks remain distinct within each branch. Each merge touches at most two blocks, allowing $\lceil t/2\rceil$ steps. Appendix~\ref{sec:lxop-recursion-details} gives the recovery argument and the resulting bounds.

\begin{lemma}\label{lem:lxop-moments}
Fix $w\geq1$ and let $\sigma$ be $w$-admissible. The fixed-support moments satisfy
\begin{equation}\label{eq:lxop-w-level-one}
    L_1=O_w(N^{-3}).
\end{equation}
There is a constant $C_w$, depending only on $w$, such that for $t\geq2$ with $(w+1)t\leq N/8$,
\begin{equation}\label{eq:lxop-w-fixed-support-simplified}
    L_t\leq\left(C_w\frac{t^{3/2}}{N^{3/2}}\right)^t.
\end{equation}
\end{lemma}
The proof is given in Appendix~\ref{sec:lxop-recursion-details}.

\paragraph{Proof of Theorem~\ref{thm:lxop-w-classical}.}
Let $\xi_q=\xi^{(\sigma,w)}_{n,q}$, and recall $L_t$ from~\eqref{eq:lxop-w-fixed-support-second-moment}. Each $L_t$ is the contribution of one fixed set of $t$ construction inputs. There are $\binom qt$ such supports among the $q$ queried inputs. Thus input symmetry and Parseval give
\begin{equation}\label{eq:lxop-w-classical-variance}
    \|\xi_q-\indicate[]\|_2^2
    =\sum_{t=1}^q\binom qt L_t.
\end{equation}
Put $a=w+1$. Lemma~\ref{lem:lxop-moments} bounds level one and, whenever $aq\leq N/8$, every level up to $q$. Using $\binom qt\leq(eq/t)^t$, we obtain
\[
    \binom qtL_t
    \leq\left(C'_w\frac{q\sqrt t}{N^{3/2}}\right)^t.
\]
The ratio of consecutive terms on the right is at most $C''_wq\sqrt{q+1}/N^{3/2}$, where $C''_w = \sqrt{e}C'_w$ absorbs the factor $((t+1)/t)^{t/2} \leq \sqrt{e}$. Choose $c_w>0$ sufficiently small so that $c_w \leq (8a)^{-1}$ and this ratio is at most $1/2$ whenever $q\leq c_wD\leq c_wN$. Such a choice is possible since $q\sqrt{q+1}/N^{3/2}\leq\sqrt2(q/N)^{3/2}$ for $q\geq1$. The same choice ensures $at\leq aq\leq N/8$ for every $t\leq q$. The contribution of all $t\geq2$ in~\eqref{eq:lxop-w-classical-variance} is then $O_w(q^2/N^3)$. Together with~\eqref{eq:lxop-w-level-one},
\[
    \|\xi_q-\indicate[]\|_2^2
    \leq O_w\left(\frac{q}{N^3}+\frac{q^2}{N^3}\right)
    =O_w\left(\frac{q^2}{N^3}\right),
\]
since $q\geq1$. Taking square roots and using statistical distance at most one half of the $L_2$ distance proves the result.\qed

\subsection{Quantum Security}\label{subsec:lxop-quantum}
The same fixed-support moments control quantum distinguishing advantage as well. We now count supports in the complete construction domain and sum only the levels up to $2q$, as required by Theorem~\ref{thm:quantum-fourier}.
\begin{corollary}\label{cor:lxop-w-quantum-fourier}
    Fix $w\geq1$. There exists a sufficiently small constant $c_w>0$, depending only on $w$, such that the following holds. Let $s=\lceil\log_2(w+1)\rceil$, let $\sigma$ be $w$-admissible, and let $\rf_w\colon\binsp[n-s]\to G^w$ be uniform. For every $0\leq q\leq c_wD$,
    \[
        \compdist{(q)}{\ket{\sop[\sigma,w]}}{\ket{\rf_w}}
        \leq O_w\left(\frac{1}{\sqrt N}\right).
    \]
\end{corollary}
\begin{proof}
    The case $q=0$ is immediate; assume $q\geq1$. Let $D=N/2^s$ be the size of the actual domain, let $\xi$ denote the density of the complete $\sop[\sigma,w]$ function, and let $L_t$ denote the fixed-support second Fourier moment from~\eqref{eq:lxop-w-fixed-support-second-moment}. By input symmetry,
    \begin{equation*}
        \sum_{|\alpha|=t}|\widehat\xi(\alpha)|^2
        =\binom Dt L_t.
    \end{equation*}
    By Lemma~\ref{lem:lxop-moments}, $L_1=O_w(N^{-3})$ and $L_t\leq(C_wt^{3/2}/N^{3/2})^t$ when $(w+1)t\leq N/8$.
    Hence, for $t\geq2$,
    \[
        \binom DtL_t
        \leq
        \left(C'_w\sqrt{\frac tN}\right)^t
        \eqqcolon A_t,
        \qquad
        \frac{A_{t+1}}{A_t}
        \leq C''_w\sqrt{\frac{t+1}{N}}.
    \]
    Shrink $c_w$ from Theorem~\ref{thm:lxop-w-classical}, if necessary, so that
    \[
        c_w\leq\min\left\{2^{-s-1},\frac{1}{16(w+1)},
                          \frac{1}{8(C''_w)^2}\right\}.
    \]
    Since $q\leq c_wD \leq c_wN$, we have $2q\leq\min\{D,N/(8(w+1))\}$, so Lemma~\ref{lem:lxop-moments} applies at every level $t \leq 2q$, and the ratio is at most $1/2$ for $2\leq t<2q$. Hence
    \[
        \sum_{t=2}^{2q}\binom DtL_t
        \leq O_w(A_2)
        =O_w(1/N),
    \]
    while the level-one contribution is $DL_1=O_w(N^{-2})$. Therefore
    \[
        \sum_{1\leq|\alpha|\leq2q}|\widehat\xi(\alpha)|^2
        \leq O_w(1/N).
    \]
    The result follows from Theorem~\ref{thm:quantum-fourier}.\qed
\end{proof}

\subsubsection{Concrete bounds for one- and two-block output.}
For $w \in \{1,2\}$, Dinur gives sharper fixed-support moments. Applying Theorem~\ref{thm:quantum-fourier} to these estimates gives the following concrete bounds. The proof is given in Appendix~\ref{sec:proof-concrete-bound-lxop}.

\begin{corollary}\label{cor:lxop-quantum-concrete}
    Let $N=2^n\geq2^{10}$, and let $\sigma$ be $1$-admissible (equivalently, $2$-admissible). For a uniform function $\rf\colon\binsp[n-1]\to G$,
    \begin{equation*}
        \compdist{(q)}{\ket{\sop[\sigma,1]}}{\ket\rf}
        \leq\sqrt{\frac{1887}{N}}
        \qquad(0\leq q\leq N/64).
    \end{equation*}
    For a uniform function $\rf_2\colon\binsp[n-2]\to G^2$,
    \begin{equation*}
        \compdist{(q)}{\ket{\sop[\sigma,2]}}{\ket{\rf_2}}
        <\sqrt{\frac{37}{2N}}
        \qquad(0\leq q\leq N/384).
    \end{equation*}
\end{corollary}
The two constants are not comparable, as they come from different estimates: for $w = 1$ we bound the level sums by a geometric series with a crude ratio, whereas for $w=2$ we evaluate the first few levels numerically and bound only the tail by a geometric series. The same numerical treatment would also reduce the constant $w=1$.

\ifanon\else
\vspace*{1cm}
{\large \noindent\textbf{Acknowledgements}}\\[0.25em]
The authors would like to thank Beno\^{i}t Cogliati and the GAPS 2025 discussion group --- Wonseok Choi, Chandranan Dhar, Ravindra Jejurikar, Jannis Leuther, Jordan Naccache, Abishanka Saha, Andr\'{e} Schrottenloher, and Rentaro Shiba --- for valuable discussions on the topic. A part of this result was conceived while A.J. was supported (in parts) by the German Research Foundation (DFG) within the framework of the Excellence Strategy of the Federal Government and the States -- EXC 2092 \textsc{CaSa} -- 390781972, and the Chair of Symmetric Cryptography at Ruhr University Bochum. This work was funded (in parts) by the German Federal Ministry of Education and Research (BMBF) under grant number 16KIS1802. The authors are responsible for the content of this publication.
\fi

\paragraph{Responsible Disclosure of AI Usage.}
We used ChatGPT 6, GPT-5.6, and Google Search prompts (powered by Gemini). AI tools were primarily used as a writing aid: to polish some passages, check the language, and check consistency throughout the paper. A secondary application was in certain numerical computations appearing in some proofs, where the human authors gave precise prompts to optimise certain functions and compute the resulting bounds, much as a numerical computing tool would be used in this context. All text and all technical claims, AI-assisted or not, were ultimately reviewed and validated by the authors, and they take full responsibility for the content of this paper.

\iflncs
\bibliographystyle{splncs04}
\else
\bibliographystyle{alphaurl}
\fi
\bibliography{references}

\ifappendix
\appendix
\renewcommand*{\theHsection}{appendix.\Alph{section}}
\renewcommand*{\theHsubsection}{appendix.\Alph{section}.\arabic{subsection}}
\section{Injection Coefficient Estimates}
\label{sec:explicit-max-moment}
The following coefficient bound is due to~\cite[Lemma~4.1]{Eberhard2017:SOP:arXiv}, written in our density normalisation.
\begin{proposition} \label{prop:eberhard-sparse-maximum}
    Let $G$ be a finite abelian group of order $N\geq2$. Then, for all $1\leq t\leq\floor{N/2}$,
    \[
        \Lambda_t\leq\binom Nt^{-1/2}.
    \]
\end{proposition}
\begin{proof}
    From~\cite[Lemma~4.1]{Eberhard2017:SOP:arXiv},
    \[
        |\findicate[\mcs](\chi)|
        \leq\binom Nt^{-1/2}\frac{N!}{N^N}
        \qquad
        (|\chi|=t\leq N/2).
    \]
    Multiplying by $N^N/N!$ gives the corresponding density coefficient, and~\eqref{eq:injection-restriction} then gives the stated bound on $\Lambda_t$.\qed
\end{proof}

\subsection{Small Supports}
At support size one, every nonprincipal character has zero coefficient, so $\Theta_1=\Lambda_1=0$. For nonprincipal $\chi_1,\chi_2\in\widehat G$,
\[
    N(N-1)\widehat{\unidens[\mci_2]}(\chi_1,\chi_2)
    =
    \sum_{y_1\neq y_2}
    \overline{\chi_1(y_1)\chi_2(y_2)}
    =
    -\sum_{y\in G}\overline{(\chi_1\chi_2)(y)}.
\]
The unrestricted sum over $G^2$ is zero, since each character is nonprincipal. Subtracting the equal-output assignments gives the last equality. By character orthogonality,
\[
    \widehat{\unidens[\mci_2]}(\chi_1,\chi_2)
    =
    \begin{cases}
        -1/(N-1), & \chi_1\chi_2=\trivchar,\\
        0, & \text{otherwise}.
    \end{cases}
\]
There are $N-1$ such ordered pairs, one for each choice of $\chi_1$, since $\chi_2=\chi_1^{-1}$ is then forced. Hence $\Lambda_2=1/(N-1)$ and $\Theta_2=1/(N-1)$.

For $N\geq3$ and three nonprincipal characters, inclusion--exclusion removes the three possible pairwise equalities. The unrestricted sum and each single-pair sum vanish, since an unrestricted nonprincipal character sums to zero. The all-equal term has coefficient two, so
\[
    (N)_3\funidens[\mci_3](\chi_1,\chi_2,\chi_3)
    =2\sum_{y\in G}\overline{(\chi_1\chi_2\chi_3)(y)}.
\]
Consequently,
\[
    \funidens[\mci_3](\chi_1,\chi_2,\chi_3)
    =\begin{cases}
        2/((N-1)(N-2)), & \chi_1\chi_2\chi_3=\trivchar,\\
        0, & \text{otherwise}.
    \end{cases}
\]
There are $(N-1)(N-2)$ such ordered triples: choose $\chi_1\neq\trivchar$, then $\chi_2\notin\{\trivchar,\chi_1^{-1}\}$, and finally $\chi_3=(\chi_1\chi_2)^{-1}$. Thus
\[
    \Lambda_3=\frac{2}{(N-1)(N-2)},
    \qquad
    \Theta_3=\frac{4}{(N-1)(N-2)}.
\]

\section{Proof of Lemma~\ref{lem:extended-injection-moments}}
\label{sec:extended-injection-proof}

We retain the notation $\Theta_t=\Theta_{t,2}$ and $\Lambda_t$ for the actual fixed-support moment and maximum. Throughout this appendix, $N\geq1024$, and $3\leq T\leq8N/15$. All logarithms in this appendix are natural.

\paragraph{Exact second moments.}
Parseval gives
\[
    \sum_{\alpha\in\dual^j}|\funidens[\mci_j](\alpha)|^2
    =\|\unidens[\mci_j]\|_2^2=\frac{N^j}{(N)_j}.
\]
Separating these characters by their support, and using~\eqref{eq:injection-restriction}, yields
\[
    \sum_{t=0}^j\binom jt\Theta_t=\frac{N^j}{(N)_j},
\]
where we take $\Theta_0=1$. Binomial inversion therefore gives the exact identity
\begin{equation}\label{eq:group-independent-weight}
    \Theta_t=\sum_{j=0}^t(-1)^{t-j}\binom tj\frac{N^j}{(N)_j}.
\end{equation}
This is the expression obtained in~\cite[Proposition~20]{Dinur2024:SOP:EUROCRYPT} in the binary setting. The preceding calculation also makes clear that it depends only on the cardinality $N$.

Let $V$ be an auxiliary real random variable with density, relative to Lebesgue measure,
\[
    p_V(v)=\frac{v^Ne^{-v}}{N!}\qquad(v>0),
    \qquad Z=\frac NV-1.
\]
For $0\leq j\leq N$, direct integration gives
\[
    \mathbf E\left[\left(\frac NV\right)^j\right]
    =\frac{N^j(N-j)!}{N!}=\frac{N^j}{(N)_j}.
\]
Consequently, the binomial-inversion identity~\eqref{eq:group-independent-weight} gives
\begin{equation}\label{eq:injection-auxiliary-moments}
    \Theta_t=\mathbf E[Z^t]\qquad(0\leq t\leq N).
\end{equation}
This variable only represents the numerical moment sequence; it is not part of the oracle distribution. Its even moments allow us to interpolate between support sizes using H\"older's inequality.

We also use the generating function
\begin{equation}\label{eq:injection-moment-generating}
    F_N(z)=(1-z)^N\exp\left(\frac{Nz}{1-z}\right),
    \qquad
    [z^t]F_N(z)=\binom Nt\Theta_t\quad(0\leq t\leq N).
\end{equation}
Indeed, substituting~\eqref{eq:group-independent-weight} and interchanging the finite sums yields
\[
    \sum_{t=0}^N\binom Nt\Theta_tz^t
    =(1-z)^N\sum_{j=0}^N\frac1{j!}
        \left(\frac{Nz}{1-z}\right)^j.
\]
The coefficients through degree $N$ agree with those of $F_N$. Moreover,
\[
    F_N(z)=\exp\left(N\sum_{j=2}^{\infty}\frac{j-1}{j}z^j\right)
\]
has non-negative coefficients. Hence, for every $0<\rho<1$, $[z^t]F_N(z)\,\rho^t\leq F_N(\rho)$, and therefore $[z^t]F_N(z)\leq \rho^{-t}F_N(\rho)$. Differentiating gives $(1-z)^2F_N'(z)=NzF_N(z)$, and comparison of coefficients gives
\[
    \Theta_{t+1}=\frac{t}{N-t}(2\Theta_t+\Theta_{t-1})
    \quad(1\leq t<N),
    \qquad \Theta_0=1,\quad\Theta_1=0.
\]
In particular,
\begin{equation}\label{eq:exact-injection-low-moments}
\begin{aligned}
    \Theta_4&=\frac{3(N+6)}{(N-1)(N-2)(N-3)},\\
    \Theta_5&=\frac{8(5N+12)}{(N-1)(N-2)(N-3)(N-4)},\\
    \Theta_6&=\frac{5(3N^2+86N+120)}{(N-1)(N-2)(N-3)(N-4)(N-5)}.
\end{aligned}
\end{equation}

\paragraph{Extending the maximum beyond $N/2$.}
To control the remaining levels, we use Eberhard's coefficient-merging identity. For $2\leq t\leq N$ and nonprincipal $\chi_1,\ldots,\chi_t\in\dual$,
write
\[
    \chi^{(i)}=(\chi_1,\ldots,\chi_i\chi_t,\ldots,\chi_{t-1}).
\]
Then using~\eqref{eq:injection-coefficient-merging}
\[
    \funidens[\mci_t](\chi_1,\ldots,\chi_t)
    =-\frac1{N-t+1}\sum_{i=1}^{t-1}
        \funidens[\mci_{t-1}](\chi^{(i)}).
\]
The merged coordinate $\chi_i\chi_t$ is either nonprincipal or principal, so the resulting character lies at level $t-1$ or $t-2$. Taking absolute values and using~\eqref{eq:injection-restriction} gives
\begin{equation}\label{eq:injection-maximum-recurrence}
    \Lambda_t\leq\frac{t-1}{N-t+1}
        \max\{\Lambda_{t-1},\Lambda_{t-2}\}
    \qquad(3\leq t\leq N).
\end{equation}

Put $m_0=\floor{N/2}$ and $C=\binom{N}{m_0}$. We apply this recurrence for $t\geq m_0+1$, starting from Proposition~\ref{prop:eberhard-sparse-maximum}, which gives
$\Lambda_t\leq\binom Nt^{-1/2}$ for $1\leq t\leq N/2$. We claim that
\begin{equation}\label{eq:extended-maximum-central}
    \Lambda_t\leq\frac{\sqrt C}{\binom Nt}
    \qquad(m_0\leq t\leq N).
\end{equation}
At $t=m_0$, this is exactly the sparse bound. At $t=m_0+1$, if $N=2m_0$, the recurrence factor is one and
\[
    \binom N{m_0-1}^{-1/2}
    =\frac1{\sqrt C}\sqrt{\frac{m_0+1}{m_0}}
    \leq\frac{\sqrt C}{\binom N{m_0+1}}.
\]
If $N=2m_0+1$, the factor is $m_0/(m_0+1)$ and
\[
    \frac{m_0}{m_0+1}\binom N{m_0-1}^{-1/2}
    =\frac1{\sqrt C}\frac{\sqrt{m_0(m_0+2)}}{m_0+1}
    \leq\frac1{\sqrt C}
    =\frac{\sqrt C}{\binom N{m_0+1}}.
\]
For $t\geq m_0+2$, the binomial coefficients are non-increasing, so the induction step follows from
\[
    \frac{t-1}{N-t+1}\frac1{\binom N{t-1}}
    =\frac{t-1}{t}\frac1{\binom Nt}
    \leq\frac1{\binom Nt}.
\]
This proves~\eqref{eq:extended-maximum-central} and extends the maximum bound beyond $N/2$.

\paragraph{The levels up to $N/2$.}
Put $L=2\floor{N/4}$, so $L$ is even and is either $m_0$ or $m_0-1$. Since the largest binomial coefficient is at least $2^N/(N+1)$,
\[
    \binom NL\geq\frac{2^N}{2(N+1)}.
\]
Using~\eqref{eq:injection-moment-generating} at $z=2/5$ therefore gives
\[
    \Theta_L\leq2(N+1)\beta^N,
    \qquad
    \beta=\frac3{10}e^{2/3}\sqrt{\frac52}<e^{-1/13}.
\]
Also, $\Theta_6\geq15/N^3$ by~\eqref{eq:exact-injection-low-moments}. For $N\geq1024$,
\begin{equation}\label{eq:injection-even-endpoint}
    \Theta_L\leq2(N+1)e^{-N/13}
    \leq\frac{15}{N^3}e^{-(L-6)/10}
    \leq\Theta_6e^{-(L-6)/10}.
\end{equation}
For the middle inequality, it suffices, using $L\leq N/2$, to check
\[
    \frac{7N}{260}-\log\left(\frac{2(N+1)N^3}{15}\right)+\frac35>0.
\]
The left-hand side is greater than $2.45727$ at $N=1024$, and its derivative is $7/260-1/(N+1)-3/N>0$ in the required range.

Both $6$ and $L$ are even. H\"older's inequality applied to~\eqref{eq:injection-auxiliary-moments}, followed by~\eqref{eq:injection-even-endpoint}, gives, for $6\leq t\leq L$,
\[
    \Theta_t\leq\mathbf E[|Z|^t]
    \leq\Theta_6^{(L-t)/(L-6)}\Theta_L^{(t-6)/(L-6)}
    \leq\Theta_6e^{-(t-6)/10}.
\]
Hence
\begin{equation}\label{eq:exact-moment-low-tail}
    \sum_{t=4}^{L}\Theta_t
    \leq\Theta_4+\Theta_5+\frac{\Theta_6}{1-e^{-1/10}}
    <\Theta_4+\Theta_5+11\Theta_6.
\end{equation}
For $4\leq t\leq L$, the sparse maximum gives
\[
    \binom Nt\Theta_{t,4}
    \leq\binom Nt\Lambda_t^2\Theta_t\leq\Theta_t.
\]
The support-three contribution is given exactly by Appendix~\ref{sec:explicit-max-moment}:
\[
    \binom N3\Theta_{3,4}
    =\frac{8N}{3(N-1)^2(N-2)^2}.
\]

\paragraph{The remaining levels.}
For $L<t\leq T$, we have $t\geq m_0$. Using~\eqref{eq:extended-maximum-central} and then~\eqref{eq:injection-moment-generating} at $z=2/5$,
\[
    \binom Nt\Theta_{t,4}
    \leq\frac{C}{\binom Nt}\Theta_t
    \leq\frac{2^N}{\binom Nt^2}
        \left(\frac35\right)^Ne^{2N/3}\left(\frac52\right)^t.
\]
For $x=t/N$, put $h(x)=-x\log x-(1-x)\log(1-x)$. The binomial probability with parameters $N,x$ is largest at $t$, and is at least $1/(N+1)$ there. Thus
\[
    \binom Nt\geq\frac{e^{Nh(x)}}{N+1}.
\]
It follows that
\[
    \binom Nt\Theta_{t,4}\leq(N+1)^2e^{NJ(x)},
    \qquad
    J(x)=\log\frac65+\frac23+x\log\frac52-2h(x).
\]
Here $x\geq m_0/N>0.49$, and $J'(x)=\log(5/2)-2\log((1-x)/x)>0$ throughout the required interval. A direct evaluation gives
\[
    J(8/15)<-0.04417000<-\frac1{23}.
\]
There are fewer than $N$ remaining levels, so
\begin{equation}\label{eq:extended-moment-high-tail}
    \sum_{L<t\leq T}\binom Nt\Theta_{t,4}
    \leq N(N+1)^2e^{-N/23}.
\end{equation}
This also covers the possible odd level just below $N/2$.

Combining~\eqref{eq:exact-moment-low-tail},~\eqref{eq:extended-moment-high-tail}, and the support-three contribution gives
\begin{equation}\label{eq:extended-moment-numerical-envelope}
    \sum_{t=3}^{T}\binom Nt\Theta_{t,4}
    <\frac{8N}{3(N-1)^2(N-2)^2}
        +\Theta_4+\Theta_5+11\Theta_6
        +N(N+1)^2e^{-N/23}.
\end{equation}
After multiplication by $N^2$, each rational term decreases with $N$: write it as a polynomial with non-negative coefficients in $1/N$, divided by products of $1-i/N$. The exponential term decreases as well, since
\[
    \frac{d}{dN}\log\left(N^3(N+1)^2e^{-N/23}\right)
    =\frac3N+\frac2{N+1}-\frac1{23}<0.
\]
At $N=1024$, the normalised right-hand side of~\eqref{eq:extended-moment-numerical-envelope} is less than $3.245641<13/4$. This proves~\eqref{eq:extended-injection-fourth-moment}.

\paragraph{A common maximum for all these levels.}
At support three,
\[
    \Lambda_3^2\binom N4
    =\frac{N(N-3)}{6(N-1)(N-2)}<1.
\]
For $4\leq t\leq m_0$, we have
\[
    \Lambda_t\leq\binom Nt^{-1/2}\leq\binom N4^{-1/2}.
\]
For $m_0<t\leq8N/15$, the entropy estimate and~\eqref{eq:extended-maximum-central} give
\[
    \Lambda_t\leq(N+1)\exp\left(N\left(\frac12\log2-h(t/N)\right)\right)
    \leq(N+1)e^{-N/3}<N^{-2}<\binom N4^{-1/2}.
\]
Indeed, $h(8/15)-\tfrac12\log2>1/3$, and $(N+1)e^{-N/3}<N^{-2}$ holds at $N=1024$ and persists since $2/N+1/(N+1)-1/3<0$. This proves~\eqref{eq:extended-injection-maximum}.

Now let $m$ be any integer with $T \leq m \leq N$. For $3\leq t\leq T\leq m\leq N$, the ratio $\binom mt/\binom Nt$ is at most $\binom m3/\binom N3$. Also, $\Theta_{t,2k}\leq\Lambda_t^{2k-4}\Theta_{t,4}$. Thus the two estimates just proved give
\[
    \sum_{t=3}^{T}\binom mt\Theta_{t,2k}
    \leq\frac{\binom m3}{\binom N3}\binom N4^{-(k-2)}
        \sum_{t=3}^{T}\binom Nt\Theta_{t,4}
    \leq\frac{\binom m3}{\binom N3}\frac{13}{4N^2}\binom N4^{-(k-2)}.
\]
This proves~\eqref{eq:common-injection-tail}.\qed

\section{Query-dependent Quantum Bounds for Sums of Permutations}
\label{sec:sop-query-dependent}
Let $F\colon\mcx\to H$ have density $\varphi$ and uniform one-point marginals, and let $\rf\colon\mcx\to H$ be uniform. For a $q$-query quantum distinguisher $\mca$ with acceptance function $a$, Lemmas~\ref{lem:layered-fourier-bound} and~\ref{lem:fourier-zhandry-method}, with $r=2$, give
\begin{equation}\label{eq:quantum-fourier-split}
    \compdist{\mca}{\ket F}{\ket\rf}
    \leq|\mathbf E[a\varphi_{=2}]|+\frac12\|R_q\|_2,
    \qquad R_q=\sum_{t=3}^{2q}\varphi_{=t},
\end{equation}
for $q\geq1$. We apply this bound to the unprojected and projected sums. Their level-two components follow from Lemma~\ref{lem:sop-two-point-law}; the common injection estimate~\eqref{eq:common-injection-tail} controls the higher levels.

\subsection{A Bound for the Centred Collision Count}
\label{subsec:collision-count}
We first bound the correlation of a quantum acceptance function with the centred collision count. We state it for an arbitrary finite input set and a finite abelian output group, so that it also applies to truncation.

\begin{lemma}\label{lem:quantum-centred-collisions}
    Let $\mcx$ be a finite set, let $H$ be a finite abelian group of order $M$, and define, for $f \colon \mcx \to H$,
    \[
        K_H(f) \coloneq \sum_{\{x,y\}\subseteq\mcx}
            \left(M\indicate[\{f(x)=f(y)\}]-1\right),
    \]
    where the sum is over unordered pairs of distinct inputs, and the expectation below is with respect to $\unimes[H^{\mcx}]$. If $a$ is the acceptance probability of a $q$-query quantum algorithm, then
    \[
        \left|\mathbf{E}_{\unimes}[aK_H]\right|
        \leq\frac{\pi^2}{6}(2q-1)^3=O(q^3)\qquad(q\geq1).
    \]
    For $q=0$, the expectation is zero. The implicit constant is independent of $\mcx$ and $H$.
\end{lemma}
Notice that $K_H$ has mean zero.
\begin{proof}
    If $q=0$, $M=1$, or $|\mcx|<2$, the claim follows from $\mathbf E[K_H]=0$ or $K_H=0$. Assume otherwise. For a positive integer $r$, sample independent uniform random functions $b \colon \mcx \to [r]$ and $c \colon [r] \to H$, and set $F_r=c\circ b$. For $z=1/r$, let $p(z)=\mathbf E[a(F_r)]$. We first make the small-range bound concrete. By Lemma~\ref{lem:fourier-zhandry-method}, the Fourier expansion of $a$ uses supports of size at most $2q$. On a support of size $s\geq1$, the labels assigned by $b$ induce a partition. A fixed partition with $j\geq1$ blocks has probability
    \[
        \frac{(r)_j}{r^s}
        =z^{s-j}\prod_{i=0}^{j-1}(1-iz),
    \]
    which has degree at most $s-1$, since the factor at $i=0$ is one. The conditional character expectation depends on the partition but not on $r$. Thus $p$ is a polynomial of degree at most $2q-1$, with $p(0)=\mathbf E_{\unimes}[a]$ and $0\leq p(1/r)\leq1$ for every positive integer $r$. Its coefficients are real, since its values at these real points are real. Zhandry's polynomial bound~\cite[Theorem~B.1, with $\Delta=1$]{Zhandry2012:SRD:FOCS} yields
    \begin{equation}\label{eq:small-range-bound}
        \left|\mathbf{E}[a(F_r)]-\mathbf{E}_{\unimes}[a]\right|
        \leq\frac{\pi^2(2q-1)^3}{6r}
        \leq C_{\mathrm{SR}}\frac{q^3}{r},
        \qquad C_{\mathrm{SR}}=\frac{4\pi^2}{3}.
    \end{equation}
    Here the support cutoff applies directly to the addition oracle on $H$.

    Let $D=|\mcx|$, and let $\varrho_r$ be the density of $F_r$ on $H^{\mcx}$. As $r\to\infty$, the probability that $b$ is injective is
    \[
        \frac{(r)_D}{r^D}=1-\frac{\binom{D}{2}}{r}+O_D(r^{-2}).
    \]
    For each fixed pair $\{x,y\}$, the probability that this is the only collision of $b$ is
    \[
        \frac{(r)_{D-1}}{r^D}=\frac{1}{r}+O_D(r^{-2}).
    \]
    Conditioned on the first event, the values of $c$ at the distinct labels chosen by $b$ are independent and uniform, so $F_r$ is a uniform random function. Conditioned on the second event, the outputs at $x$ and $y$ are equal, and all other outputs remain independent and uniform. The density of this distribution is $M\indicate[\{f(x)=f(y)\}]$. All other collision patterns have total probability $O_D(r^{-2})$. Consequently,
    \[
        \varrho_r = \indicate[]+\frac{K_H}{r}+O_{D,M}(r^{-2}),
    \]
    where the error is uniform on the finite space $H^{\mcx}$, with $\mcx$ and $H$ fixed. Thus,
    \[
        \lim_{r\to\infty}r\left(\mathbf{E}[a(F_r)]-\mathbf{E}_{\unimes}[a]\right)
        =\mathbf{E}_{\unimes}[aK_H].
    \]
    Multiplying~\eqref{eq:small-range-bound} by $r$ and taking the limit proves the claim. The limit is taken with $\mcx,H$ and the algorithm fixed. Although the vanishing error can depend on $D,M$, the surviving constant $C_{\mathrm{SR}}$ does not.\qed
\end{proof}

\subsection{Proof of Theorem~\ref{thm:sop-quantum-combined-local}}
Fix $k\geq2$, and let $\varphi_k=\unidens[\mcs]^{*k}$ be the density of the complete table of $\sop[k]$. Its one-point marginals are uniform. Since the two-point marginals determine the level-two projection, Lemma~\ref{lem:sop-two-point-law} gives
\begin{equation}\label{eq:sop-k-collision-part}
    h_{2,k}\coloneq(\varphi_k)_{=2}
    =\eta_k K_G,\qquad\eta_k=\left(-\frac1{N-1}\right)^k.
\end{equation}
Thus Lemma~\ref{lem:quantum-centred-collisions} implies, for every $q$-query acceptance function $a$,
\[
    |\mathbf E[ah_{2,k}]|
    \leq\frac{\pi^2(2q-1)^3}{6(N-1)^k}.
\]
For $q=1$, there is no remainder, proving the one-query bounds.

Suppose $N\geq1024$ and $2\leq q\leq\sqrt N$. By Parseval's identity, convolution-multiplication duality, and~\eqref{eq:common-injection-tail} with $m=N$ and $T=2q$,
\[
    \|R_{q,k}\|_2^2
    =\sum_{t=3}^{2q}\binom Nt\Theta_{t,2k}
    \leq\frac{13}{4N^2}\binom N4^{-(k-2)},
    \qquad R_{q,k}=\sum_{t=3}^{2q}(\varphi_k)_{=t}.
\]
The remainder in~\eqref{eq:quantum-fourier-split} is therefore $O_k(N^{-2k+3})$. For $k=2$, this gives $O(q^3/N^2+1/N)$; for $k\geq3$, it gives $O_k(q^3/N^k)$. The finitely many orders $N<1024$ are absorbed into the implicit constants.

Corollary~\ref{cor:sop-quantum-completion} gives the uniform bound $O_k(N^{-(k-3/2)})$ throughout $2q<N$. For $q>\sqrt N$, this is no larger than $q^3/N^k$. Combining the local and uniform bounds proves the theorem.\qed

Keeping the constants in the preceding calculation gives the following bound for $q < 4N/15$.
\begin{corollary}\label{cor:sop-query-dependent-concrete}
    Let $G$ be a finite abelian group of order $N\geq1024$, fix $k\geq2$, and let $\rf\colon G\to G$ be uniform. For $2\leq q\leq4N/15$,
    \begin{equation}\label{eq:sop-query-dependent-concrete}
    \begin{split}
        \compdist{(q)}{\ket{\sop}}{\ket\rf}
        \leq\min\Bigg\{&
            \sqrt{\frac{N}{8(N-1)^{2k-2}}
                +\frac{13}{16N^2}\binom N4^{-(k-2)}},\\
            &\frac{\pi^2(2q-1)^3}{6(N-1)^k}
                +\frac{\sqrt{13}}{4N}\binom N4^{-(k-2)/2}
        \Bigg\}.
    \end{split}
    \end{equation}
    For $q=1$, the advantage is at most $\pi^2/(6(N-1)^k)$; for $q=0$, it is zero.
\end{corollary}
\begin{proof}
    The first bound is Corollary~\ref{cor:sop-quantum-concrete}. For the second, the two-input identity~\eqref{eq:sop-k-collision-part}, which also holds for $k=2$, and Lemma~\ref{lem:quantum-centred-collisions} give
    \[
        |\mathbf E_{\unimes}[ah_{2,k}]|
        \leq\frac{\pi^2(2q-1)^3}{6(N-1)^k}.
    \]
    For the remaining levels,~\eqref{eq:common-injection-tail} gives
    \[
        \|R_{q,k}\|_2^2
        =\sum_{t=3}^{2q}\binom Nt\Theta_{t,2k}
        \leq\frac{13}{4N^2}\binom N4^{-(k-2)}.
    \]
    Apply~\eqref{eq:quantum-fourier-split}, retaining its factor $1/2$ on the remainder. For $q=1$, the remainder is empty, and for $q=0$ the oracles are not queried.\qed
\end{proof}

\subsection{Proof of Theorem~\ref{thm:generalised-sop-quantum}}
\label{subsec:truncated-sop-query-dependent}
The same decomposition gives a bound that retains the output-size dependence for the construction from Section~\ref{subsec:trunc-sop}. The additional observation is that the proof of Proposition~\ref{prop:subspace-moment-bound} retains the rank of the span of the character coordinates. Keeping this rank information gives a finer estimate on the remainder than applying the proposition directly.

Throughout this subsection, take $G=\binsp$, $H=\binsp[m]$, $N=2^n$, and $M=2^m$, with $0\leq m\leq n$. Let $\tau\colon G\to H$ be surjective and linear.
\paragraph{Two permutations.}
    For $M=1$, both oracles are the unique function into the trivial group. Assume $M\geq2$. As before, the finitely many orders $N<1024$ can be absorbed into the asymptotic constant. Fix a $q$-query distinguisher $\mca$ with acceptance probability $a$. Let $\nu$ be the density of $\tau\sop[2]$ on $H^G$, and set $W=\tau^\top(H)$. As in the proof of Theorem~\ref{thm:generalised-sop-classical},
    \begin{equation}\label{eq:truncated-full-table-coefficients}
        \widehat{\nu}(\eta)=\funidens[\mcs](\tau^\top\eta)^2.
    \end{equation}
    Here $\tau^\top$ acts coordinatewise and is injective, so it preserves support size. Write
    \[
        \Pi_{\leq2q}(\nu-\indicate[])=h_2+R_q,
        \qquad h_2=\nu_{=2},
        \qquad R_q=\sum_{t=3}^{2q}\nu_{=t}.
    \]
    The pushforward identity~\eqref{eq:sop-two-point-pushforward} gives
    \[
        h_2(f)=\frac{K_H(f)}{(N-1)^2},
        \qquad
        \left|\mathbf{E}_{\unimes}[ah_2]\right|
        \leq O(q^3/N^2),
    \]
    where the second inequality follows from Lemma~\ref{lem:quantum-centred-collisions}.

    For $q=1$, $R_q=0$. For $2\leq q\leq\sqrt N$, Parseval and~\eqref{eq:truncated-full-table-coefficients} give
    \begin{equation}\label{eq:truncated-remainder-mass}
        \|R_q\|_2^2
        =\sum_{\substack{\chi\in W^G\\3\leq|\chi|\leq2q}}
            |\funidens[\mcs](\chi)|^4.
    \end{equation}
    The subspace-averaging identity~\eqref{eq:subspace-main-bound}, restricted to the support levels in~\eqref{eq:truncated-remainder-mass}, gives
    \begin{equation}\label{eq:truncated-rank-averaging}
        \|R_q\|_2^2
        =\sum_{3\leq|\chi|\leq2q}
            p_{n,m}(r(\chi))|\funidens[\mcs](\chi)|^4.
    \end{equation}
    Here $p_{n,m}$ is the same probability function as in~\eqref{eq:subspace-membership-probability}. The restriction to these levels is allowed because the diagonal action of $\glbinsp$ preserves support size. In particular,
    \[
        p_{n,m}(1)=\frac{M-1}{N-1},
        \qquad
        p_{n,m}(r)\leq p_{n,m}(2)=\frac{(M-1)(M-2)}{(N-1)(N-2)}
        \quad(r\geq2).
    \]

    A rank-one tuple over $\field$ has the same nonzero character at every position in its support. Its permutation coefficient vanishes when the support size is odd, by translating all outputs by a vector on which this character takes value $-1$. There are $(N-1)\binom{N}{t}$ rank-one tuples of support size $t$. For $4\leq t\leq2q\leq2\sqrt N<N/2$, Proposition~\ref{prop:eberhard-sparse-maximum} gives $\Lambda_t\leq\binom Nt^{-1/2}$. Thus,
    \begin{equation}\label{eq:truncated-rank-one-mass}
    \begin{aligned}
        \sum_{\substack{3\leq|\chi|\leq2q\\r(\chi)=1}}
            |\funidens[\mcs](\chi)|^4
        &\leq(N-1)\sum_{t=4}^{2q}\binom{N}{t}^{-1}\\
        &=O(N^{-3}).
    \end{aligned}
    \end{equation}
    For the last step, the successive ratio of $\binom{N}{t}^{-1}$ is $(t+1)/(N-t)=O(N^{-1/2})$ throughout $4\leq t<2q\leq2\sqrt{N}$. Here $N\geq1024$ suffices.

    Applying~\eqref{eq:truncated-rank-one-mass} to the rank-one part of~\eqref{eq:truncated-rank-averaging}, and~\eqref{eq:common-injection-tail}, with $k=2$ and $m=N$, to the remaining part, gives
    \begin{align*}
        \|R_q\|_2^2
        &\leq p_{n,m}(1)O(N^{-3})+p_{n,m}(2)O(N^{-2})\\
        &\leq O\left(\frac{M}{N^4}+\frac{M^2}{N^4}\right)
         =O(M^2/N^4).
    \end{align*}
    Hence $\|R_q\|_2=O(M/N^2)$. Now,~\eqref{eq:quantum-fourier-split} yields
    \begin{equation}\label{eq:truncated-local-query-bound}
        \compdist{(q)}{\ket{\tau\sop[2]}}{\ket{\rf}}
        \leq O\left(\frac{q^3+M}{N^2}\right)
        \qquad(q\leq\sqrt{N}).
    \end{equation}
    Again, $R_q=0$ for $q=1$.

    Finally, Corollary~\ref{cor:truncated-sop-quantum-completion} with $k=2$ gives
    \[
        \compdist{(q)}{\ket{\tau\sop[2]}}{\ket{\rf}}
        \leq O(\sqrt{M}/N)\qquad(2q<N).
    \]
    Combining this with~\eqref{eq:truncated-local-query-bound} proves~\eqref{eq:truncated-sop-query-dependent}. For $q>\sqrt{N}$, the minimum selects the uniform bound, since $q^3/N^2>N^{-1/2}\geq\sqrt{M}/N$.

For $M^{1/3}\leq q\leq\sqrt{N}$, the local bound~\eqref{eq:truncated-local-query-bound} is $O(q^3/N^2)$. In particular, for a fixed-size output group, this holds throughout $1\leq q\leq\sqrt{N}$, with the implicit constant allowed to depend on that fixed size. The improvement from $1/N$ to $M/N^2$ in the remainder is a consequence of separating rank-one characters from the characters of rank at least two.

\paragraph{Three or more permutations.}
    The case $M=1$ is immediate, and the finitely many orders $N<1024$ can be absorbed into the constant depending on $k$. Fix a $q$-query distinguisher $\mca$ with acceptance probability $a$, and let $\nu_k$ be the density of $\tau\sop[k]$ on $H^G$. For $q\leq\sqrt N$, split its low-support part into $h_{2,k}+R_{q,k}$. By~\eqref{eq:sop-two-point-pushforward},
    \[
        h_{2,k}(f)=\frac{(-1)^k}{(N-1)^k}K_H(f),
    \]
    and hence Lemma~\ref{lem:quantum-centred-collisions} gives a contribution $O_k(q^3/N^k)$.

    If $q=1$, the remainder is zero. Otherwise, the pullback formula~\eqref{eq:truncated-full-table-coefficients} and the convolution identity show that each coefficient of the $k$-sum is the corresponding coefficient of the two-sum multiplied by $\funidens[\mcs](\chi)^{k-2}$. As above, Lemma~\ref{lem:extended-injection-moments} gives
    \[
        \max_{3\leq|\chi|\leq2q}|\funidens[\mcs](\chi)|
        \leq\binom N4^{-1/2}=O(N^{-2})
        \qquad(q\leq\sqrt N).
    \]
    Together with the estimate $\|R_{q,2}\|_2=O(M/N^2)$ proved above, this gives
    \[
        \|R_{q,k}\|_2
        \leq O_k(N^{-2k+4})\frac{M}{N^2}
        =O_k\left(\frac{M}{N^{2k-2}}\right).
    \]
    Since $M\leq N$ and $k\geq3$, this is at most $O_k(N^{-k})$, and therefore at most $O_k(q^3/N^k)$. This proves the local bound $O_k(q^3/N^k)$ for $q\leq\sqrt N$. Corollary~\ref{cor:truncated-sop-quantum-completion} gives the uniform bound $O_k(\sqrt M/N^{k-1})$. For $q>\sqrt N$, the uniform bound is no larger than $N^{-(k-3/2)}<q^3/N^k$, and the result follows.\qed

\begin{remark}
    The query-dependent refinement uses the exact centred-collision form of the support-two component for sums of independent permutations. For $\mathsf{LXoP}$ this component has a different linear structure, so the same refinement does not follow from the preceding argument.
\end{remark}

\section{A Deutsch-Type One-Query Algorithm}
\label{sec:sop-two-point-attack}

We give the standard two-point phase-query algorithm with an arbitrary
nonzero output mask; see also Bonnetain et al.~\cite{BonnetainLNS2021:Linearization:ASIACRYPT}.
Let $G=\mathbb F_2^n$ and let
\[
    U_f\ket{u}_X\ket{v}_Y
    =\ket{u}_X\ket{v\oplus f(u)}_Y
\]
be the standard quantum oracle.

For distinct $x,x'\in G$ and nonzero $\lambda\in G$, there exists a
one-query quantum algorithm that returns
$b_\lambda(x,x')=\langle\lambda,f(x)\oplus f(x')\rangle$ with
certainty.

The algorithm uses a one-qubit register $B$ and two $n$-qubit
registers $X$ and $Y$.

\begin{framed}
\small
\noindent\textbf{Algorithm.}
\begin{enumerate}[label=\arabic*.,leftmargin=*,nosep]
    \item Prepare
    $\ket{0}_B\ket{0^n}_X\ket{\lambda}_Y$.
    \item Apply $H$ to $B$ and $H^{\otimes n}$ to $Y$.
    \item XOR $x$ into $X$ if $B=0$, and XOR $x'$ into $X$ if $B=1$.
    \item Apply $U_f$ once to $X$ and $Y$.
    \item Undo Step~3.
    \item Apply $H$ to $B$, measure it, and return the outcome.
\end{enumerate}
\end{framed}
Adding $z\in G$ to $H^{\otimes n}\ket{\lambda}$ contributes the phase
$(-1)^{\langle\lambda,z\rangle}$. Hence, after the oracle query and
Step~5, register $B$ is in the state
\[
\begin{aligned}
&\frac{
(-1)^{\langle\lambda,f(x)\rangle}\ket{0}
+
(-1)^{\langle\lambda,f(x')\rangle}\ket{1}
}{\sqrt2}                                                    \\[2mm]
&\qquad =
(-1)^{\langle\lambda,f(x)\rangle}
\frac{
\ket{0}+(-1)^{b_\lambda(x,x')}\ket{1}
}{\sqrt2}
=
(-1)^{\langle\lambda,f(x)\rangle}
H\ket{b_\lambda(x,x')}.
\end{aligned}
\]
Since $H^2=I$, the final Hadamard transform leaves $B$ in
$\ket{b_\lambda(x,x')}$ up to a global phase. Measuring $B$ returns
the required value with certainty. Only Step~4 queries $U_f$.\qed

\section{The Recursive Moment Argument for $\mathsf{LXoP}$}
\label{sec:lxop-recursion-details}
Throughout this appendix, $G=\binsp$ and $N=2^n$. We prove Lemma~\ref{lem:lxop-moments}. The argument follows Dinur's recursive moment bound~\cite{Dinur2025:More-SOP:EUROCRYPT}, beginning with~\eqref{eq:injection-coefficient-merging} and keeping track of a family of masks when bounding its second Fourier moment. We first prove the following more precise estimate for $t\geq2$.

\begin{proposition}\label{prop:lxop-w-fixed-support-moment}
    Fix $w\geq1$, let $\sigma$ be $w$-admissible, and let $t\geq2$
    satisfy $(w+1)t\leq N/8$. Put
    $d_t=\lceil t/2\rceil$. Then
    \begin{equation}\label{eq:lxop-w-fixed-support-moment}
        L_t
        \leq
        2^{d_t+(w+1)t}
        \left(
            \frac{(w+1)t}{N-(w+1)t}
        \right)^{t+d_t}.
    \end{equation}
\end{proposition}

We use Dinur's fixed-support second-moment estimate~\cite[Lemma~2]{Dinur2025:More-SOP:EUROCRYPT},
\begin{equation}\label{eq:dinur-fixed-support-second-moment}
    \Theta_k\leq\left(\frac{k}{N-k}\right)^{k/2},
    \qquad 1\leq k\leq N/2.
\end{equation}

\subsection{From Coefficient Merging to the Moment Bound}
\label{subsec:lxop-recursive-moment-proof}
Let $S$ be a finite set and $T\colon S\to G^r$, with $1\leq r\leq N$, be injective. Suppose the masks $T(\alpha)$ have the same support, of size $k_0$, and fix an integer $d\geq0$ with $2d<k_0\leq N/8$. At each nonempty internal node, the primary coordinate is chosen from the common support using only the initial data and the previous merge indices and zero/nonzero outcomes. We assume that, for every secondary choice and each nonempty outcome, the removed entry is uniquely determined by the child mask and these known indices and outcomes. We derive the bound under this assumption, and verify it for the LXoP family in Appendix~\ref{subsec:lxop-merge-recovery}.

A current mask is a tuple $\beta\in G^r$. For distinct indices $j,i$ in its support, define the merged mask by
\[
    \beta^{(j,i)}_\ell=
    \begin{cases}
        0,&\ell=j,\\
        \beta_i+\beta_j,&\ell=i,\\
        \beta_\ell,&\ell\notin\{i,j\}.
    \end{cases}
\]
We keep the zero coordinates rather than delete them, so the positions retain their original names. This does not affect the Fourier coefficient, by~\eqref{eq:injection-restriction}. If the support of $\beta$ is $J$, with $|J|=k\geq2$, Equation~\eqref{eq:injection-coefficient-merging}, applied to the nonzero entries and with the primary entry placed last, gives
\[
    \funidens[\mci_r](\beta)
    =-\frac1{N-k+1}\sum_{i\in J\setminus\{j\}}
        \funidens[\mci_r](\beta^{(j,i)}).
\]
This is exactly the identity used in~\cite[Proposition~8]{Dinur2025:More-SOP:EUROCRYPT}, attributed there to~\cite[Section~4]{Eberhard2017:SOP:arXiv}. Taking absolute values gives the maximum recurrence used in Appendix~\ref{sec:extended-injection-proof}. For the present purpose, we instead square the identity and apply Cauchy--Schwarz:
\begin{equation}\label{eq:lxop-moment-one-merge}
    |\funidens[\mci_r](\beta)|^2
    \leq\frac{k-1}{(N-k+1)^2}
        \sum_{i\in J\setminus\{j\}}
            |\funidens[\mci_r](\beta^{(j,i)})|^2.
\end{equation}
This is the one-step estimate of~\cite[Proposition~9]{Dinur2025:More-SOP:EUROCRYPT}. The distinction is that we will sum it over the structured family, rather than replace each term by a maximum coefficient.

\paragraph{A recursion node.}
At a node $v$, let $S_v\subseteq S$ be the set of original masks still under consideration, and let $T_v$ map them to their current masks. Write $J_v$ for their common support and $k_v=|J_v|$. At the root, these are $S,T$, and the initial support of size $k_0$. Choose one primary index $j\in J_v$ for this entire family, using the initial data and the merge history as above. Before depth $d$, we have $k_v\geq k_0-2(d-1)>2$, so a secondary index is available. For each secondary index $i\in J_v\setminus\{j\}$, put
\[
    T_{v,i}(\alpha)=T_v(\alpha)^{(j,i)},
\]
and split $S_v$ into
\[
\begin{aligned}
    S_{v,i,0}&=\{\alpha\in S_v:T_{v,i}(\alpha)_i=0\},\\
    S_{v,i,1}&=\{\alpha\in S_v:T_{v,i}(\alpha)_i\neq0\}.
\end{aligned}
\]
The zero branch has common support $J_v\setminus\{j,i\}$, while the nonzero branch has common support $J_v\setminus\{j\}$. Thus the support remains fixed within each child, even though the two children lie at different Fourier levels. Summing~\eqref{eq:lxop-moment-one-merge} gives
\begin{equation}\label{eq:lxop-moment-node-bound}
\begin{aligned}
    &\sum_{\alpha\in S_v}|\funidens[\mci_r](T_v(\alpha))|^2\\
    &\quad\leq\frac{k_v-1}{(N-k_v+1)^2}
        \sum_{i\in J_v\setminus\{j\}}\sum_{b\in\{0,1\}}
        \sum_{\alpha\in S_{v,i,b}}
            |\funidens[\mci_r](T_{v,i}(\alpha))|^2.
\end{aligned}
\end{equation}
For a fixed $i$, the two branches partition $S_v$. Different choices of $i$ need not give disjoint images; each such choice is already a separate term in~\eqref{eq:lxop-moment-node-bound}.

\paragraph{Why recovery is needed.}
By the recovery assumption, the removed primary entry $\beta_j$ is determined by the child mask and the known branch. The secondary value is then recovered from
\[
    \beta_i=(\beta_i+\beta_j)+\beta_j,
\]
and the other coordinates have not changed. Hence the full parent mask is determined. Starting from the injective map $T$, this proves inductively that the map at each node is injective. This is Dinur's applicability condition~\cite[Proposition~11]{Dinur2025:More-SOP:EUROCRYPT}. It cannot be dropped, since initial injectivity alone does not imply injectivity after a merge. For example, all masks of the form $(\gamma,\gamma)$ collapse to $(0,0)$ when their two coordinates are merged.

At a leaf $v$, trim the common zero coordinates. The resulting masks are distinct elements of $(G\setminus\{0\})^{k_v}$, so
\begin{equation}\label{eq:lxop-recursion-leaf-moment}
    \sum_{\alpha\in S_v}|\funidens[\mci_r](T_v(\alpha))|^2
    \leq\Theta_{k_v}.
\end{equation}
No factor $\binom r{k_v}$ is needed here, because the leaf has one fixed support. Without recovery, a trimmed mask could occur repeatedly on the left, and the comparison with $\Theta_{k_v}$ would not follow.

\paragraph{Summing the recursion.}
We stop at depth $d$ and apply the fixed-support estimate~\eqref{eq:dinur-fixed-support-second-moment} at each leaf. After $d$ merges, every nonempty leaf has
\[
    k_0-2d\leq k_v\leq k_0-d.
\]
In particular, $k_v>0$ by $2d<k_0$, and $k_v\leq k_0\leq N/8$. Since $k_0/(N-k_0)<1$, the leaf contribution is at most
\[
    \Theta_{k_v}
    \leq\left(\frac{k_v}{N-k_v}\right)^{k_v/2}
    \leq\left(\frac{k_0}{N-k_0}\right)^{k_0/2-d}.
\]
Each node has at most $2(k_v-1)\leq2k_0$ children, and the factor preceding the child sums in~\eqref{eq:lxop-moment-node-bound} is at most $k_0/(N-k_0)^2$. There are therefore at most $(2k_0)^d$ nonempty leaves, with exactly $d$ such factors along each root-to-leaf path. Iterating~\eqref{eq:lxop-moment-node-bound} and then applying~\eqref{eq:lxop-recursion-leaf-moment} gives
\begin{equation}\label{eq:dinur-recursive-moment}
\begin{aligned}
    \sum_{\alpha\in S}|\funidens[\mci_r](T(\alpha))|^2
    &\leq(2k_0)^d\left(\frac{k_0}{(N-k_0)^2}\right)^d
        \left(\frac{k_0}{N-k_0}\right)^{k_0/2-d}\\
    &=2^d\left(\frac{k_0}{N-k_0}\right)^{k_0/2+d}.
\end{aligned}
\end{equation}
For $d=0$, the same calculation is just the leaf estimate. This is the weaker of the two bounds in~\cite[Lemma~3]{Dinur2025:More-SOP:EUROCRYPT}, which is sufficient for the general-$w$ analysis. Relative to the direct bound $\Theta_{k_0}\leq(k_0/(N-k_0))^{k_0/2}$, the right-hand side gains a factor $(2k_0/(N-k_0))^d$. The purpose of the block relations is to make this recursion possible for sufficiently many steps.

\subsection{Recovering Labels from an Untouched Block}
\label{subsec:lxop-merge-recovery}
We retain $C_j=A^{w-j}$ from~\eqref{eq:lxop-w-parity-check}. At a fixed recursion node, choose a primary coordinate in an untouched block, meaning that none of its entries has been selected as primary or secondary in the previous merges. Suppressing the block index, write $\beta=(\beta_0,\ldots,\beta_w)$ for its entries before the merge, and $\beta'_\ell$ for its entries after the merge. The indices in the following two calculations are local to this block. The relation
\[
    \sum_{\ell=0}^w C_\ell\beta_\ell=0
\]
is available because the block is untouched. We do not assume that every block still satisfies this relation after a merge. In fact, the merged masks need not remain in the image of $T_{\sigma,w}$.

\paragraph{The secondary coordinate is in another block.}
Only the primary coordinate $j$ changes within the chosen block. Thus $\beta'_\ell=\beta_\ell$ for $\ell\neq j$, and
\[
    \beta_j=C_j^{-1}\sum_{\ell\neq j}C_\ell\beta'_\ell.
\]
The secondary value in the other block is then recovered by adding $\beta_j$ to its post-merge value. No relation in that other block is needed.

\paragraph{The secondary coordinate is in the same block.}
Let $i\neq j$ be its index. Since $\beta'_i=\beta_i+\beta_j$, substitution in the block relation gives
\[
    (C_i+C_j)\beta_j
    =C_i\beta'_i+\sum_{\ell\notin\{i,j\}}C_\ell\beta'_\ell.
\]
Now
\[
    C_i+C_j
    =A^{w-\max\{i,j\}}(I+A^{|i-j|})
\]
is invertible: $A=\sigma^\top$ is invertible, and $I+A^{|i-j|}=(I+\sigma^{|i-j|})^\top$ is invertible by $w$-admissibility. Consequently,
\[
    \beta_j=(C_i+C_j)^{-1}
       \left(C_i\beta'_i+\sum_{\ell\notin\{i,j\}}C_\ell\beta'_\ell\right),
    \qquad
    \beta_i=\beta'_i+\beta_j.
\]
This also covers $\beta'_i=0$. The zero value and its coordinate position are retained, and do not remove any information needed by the recovery formula. These formulas determine the parent mask within each branch, establishing the recovery assumption used to derive~\eqref{eq:dinur-recursive-moment}.

\paragraph{An example with two output blocks.}
For $w=2$, one block satisfies
\[
    A^2\beta_0+A\beta_1+\beta_2=0.
\]
If coordinate $0$ is primary and coordinate $1$ is secondary, the child block is $(0,\gamma,\delta)$, where $\gamma=\beta_0+\beta_1$ and $\delta=\beta_2$. The parent is recovered by
\[
    \beta_0=(A^2+A)^{-1}(A\gamma+\delta),
    \qquad \beta_1=\gamma+\beta_0,
    \qquad \beta_2=\delta.
\]
The formula applies equally when $\gamma=0$. Here $A^2+A=A(I+A)$ is invertible. This is precisely where the linear combining map is needed: it lets us reverse a merge that would otherwise lose one character label.

\paragraph{Proof of Proposition~\ref{prop:lxop-w-fixed-support-moment}.}
Fix $t\geq2$ with $(w+1)t\leq N/8$, and put $r=(w+1)t$ and $d_t=\ceil{t/2}$. Let
\[
    S_t=\{\alpha\in(G^w)^t:\alpha^{(i)}\neq0\text{ for all }i\in[t]\}.
\]
Partition $S_t$ according to the exact support of $T_{\sigma,w}(\alpha)$. There are at most $2^r$ support patterns. Fix a nonempty part $S$ with injection support $J_0\subseteq[r]$, and let $T$ be the restriction of $T_{\sigma,w}$ to $S$. The map $T$ is injective, and all its masks have the same support, of size $k_0=|J_0|$. Each nonzero construction block gives between $2$ and $w+1$ nonzero injection entries, so $2t\leq k_0\leq(w+1)t\leq N/8$. Also $2d_t\leq t+1<2t\leq k_0$, since $t\geq2$.

Each merge touches at most two original blocks. At depth $\ell<d_t$, at least $t-2\ell>0$ blocks have not been touched. Their original entries are still present, and each such block has a nonzero entry. Choose the first untouched block and its first nonzero coordinate as primary. This choice is the same for all masks at the node and is determined by the initial support and the previous merge indices and outcomes, not by the values of the entries. The preceding recovery formulas apply for every secondary choice, including when the merged entry is zero. Thus each child map remains injective, and the recursion can be continued to depth $d_t$. No primary selection is needed at a leaf.

Applying~\eqref{eq:dinur-recursive-moment} to this fixed part gives
\[
\begin{aligned}
    \sum_{\alpha\in S}
        \left|\funidens[\mci_r](T_{\sigma,w}(\alpha))\right|^2
    &\leq 2^{d_t}\left(\frac{k_0}{N-k_0}\right)^{k_0/2+d_t}\\
    &\leq 2^{d_t}\left(\frac{(w+1)t}{N-(w+1)t}\right)^{t+d_t}.
\end{aligned}
\]
For the last inequality, the ratio is at most $((w+1)t)/(N-(w+1)t)<1$, while $k_0/2\geq t$. Summing over at most $2^{(w+1)t}$ support patterns and using the pullback equality in~\eqref{eq:lxop-w-fixed-support-second-moment} proves~\eqref{eq:lxop-w-fixed-support-moment}. To obtain~\eqref{eq:lxop-w-fixed-support-simplified}, note that $d_t\leq3t/4$, $t+d_t\geq3t/2$, and $N-(w+1)t\geq7N/8$. Since $(w+1)t/(N-(w+1)t)<1$, the extra factors are absorbed into a constant $C_w$ raised to the power $t$.\qed

\subsection{The One-Input Moment}
\label{subsec:lxop-one-input-moment}
We give the remaining calculation for~\eqref{eq:lxop-w-level-one}. Put $a=w+1$ and first assume $N\geq8a$. For this single construction block, suppress the block index and partition the nonzero masks $\alpha\in G^w$ according to the support of their pullback $\beta=T_{\sigma,w}(\alpha)$. Each support has size between $2$ and $a$. Translation invariance of the injection distribution, as recorded in the preliminaries, gives
\[
    \funidens[\mci_a](\beta)=0
    \quad\text{unless}\quad
    \sum_{j=0}^w\beta_j=0.
\]
For a pullback of support size two, at positions $i,j$, this requires $\beta_i=\beta_j$. The block relation would then give
\[
    (C_i+C_j)\beta_i=0,
\]
contradicting $w$-admissibility and $\beta_i\neq0$. Thus these coefficients vanish.

Consider next a fixed support of size three, at positions $i,j,\ell$. A nonzero coefficient requires $\beta_\ell=\beta_i+\beta_j$. Substitution in~\eqref{eq:lxop-w-parity-check} gives
\[
    (C_i+C_\ell)\beta_i+(C_j+C_\ell)\beta_j=0.
\]
Both matrices are invertible. After choosing the nonzero value $\beta_i$, the values $\beta_j$ and $\beta_\ell$ are determined. Hence there are at most $N-1$ contributing masks for this support. By Appendix~\ref{sec:explicit-max-moment}, their nonzero coefficient has magnitude $2/((N-1)(N-2))$. Their total second moment is therefore at most
\[
    (N-1)\left(\frac{2}{(N-1)(N-2)}\right)^2
    =\frac{4}{(N-1)(N-2)^2}
    =O(N^{-3}).
\]
The injectivity of $T_{\sigma,w}$ ensures that each such injection mask corresponds to at most one construction mask.

For a fixed pullback support of size $k_0\geq4$, apply the recursive estimate~\eqref{eq:dinur-recursive-moment} with $r=a$ and $d=1$. Indeed, $2<k_0\leq a\leq N/8$, and the sole original block is untouched before that merge, so the recovery argument of Appendix~\ref{subsec:lxop-merge-recovery} applies. The contribution is at most
\[
    2\left(\frac{k_0}{N-k_0}\right)^{k_0/2+1}
    =O_w(N^{-3}).
\]
Since $w$ is fixed, there are only $O_w(1)$ support patterns. Summing their contributions gives $L_1=O_w(N^{-3})$. The finitely many $N<8a$ are absorbed into the implicit constant. All these estimates are uniform over the $w$-admissible maps under consideration.

\section{Proofs of Concrete Bounds for Dinur's $\mathsf{LXoP}$}
\label{sec:proof-concrete-bound-lxop}
We prove the two bounds in Corollary~\ref{cor:lxop-quantum-concrete}.
\subsection{One-Block Output}
\label{subsec:proof-concrete-quantum-lxop}
The case $q=0$ is immediate; assume $q\geq1$. Let $\xi$ denote the density of the complete $\sop[\sigma,1]$ function on its actual domain, of size $D=N/2$. By input symmetry and the definition of $\xi^{(\sigma,1)}_{n,t}$,
\begin{equation*}
    \sum_{|\alpha|=t}|\widehat\xi(\alpha)|^2
    =\binom Dt
    \sum_{\alpha\in(G\setminus\{0\})^t}
    |\widehat{\xi^{(\sigma,1)}_{n,t}}(\alpha)|^2.
\end{equation*}
The $t=1$ term is zero by the support-two calculation in Appendix~\ref{subsec:lxop-one-input-moment}, as also shown in~\cite[Proposition~15]{Dinur2025:More-SOP:EUROCRYPT}. For $t\geq2$,~\cite[Proposition~16]{Dinur2025:More-SOP:EUROCRYPT} gives
\begin{equation*}
    \sum_{\alpha\in(G\setminus\{0\})^t}
    |\widehat{\xi^{(\sigma,1)}_{n,t}}(\alpha)|^2
    \leq
    \frac{2^{3t/2+3c_t}t^{t+2c_t}(t-2c_t)^{t/2-c_t}}
    {(N-2t)^{3t/2+c_t}},
\end{equation*}
where $c_t=0$ for even $t$ and $c_t=1/2$ for odd $t$. Proposition~16 applies for $2\leq t\leq N/16$, which contains every level used below. Suppose $2\leq t\leq N/32$. Using $D=N/2$, $\binom Dt\leq(eD/t)^t$, $N-2t\geq15N/16$, and $c_t\leq1/2$, we obtain
\[
    \sum_{|\alpha|=t}|\widehat\xi(\alpha)|^2
    \leq
    2^{3/2}t
    \left(
        e\sqrt2\left(\frac{16}{15}\right)^{3/2}
        \sqrt{\frac tN}
    \right)^t.
\]
Since $e^2(16/15)^3<9$ and $t/N\leq1/32$,
\[
    \sum_{|\alpha|=t}|\widehat\xi(\alpha)|^2
    \leq\frac{51}{N}t^2\left(\frac34\right)^{t-2}.
\]
Therefore, for $q\leq N/64$,
\[
    \sum_{1\leq|\alpha|\leq2q}|\widehat\xi(\alpha)|^2
    \leq\frac{51}{N}\sum_{t=2}^\infty t^2\left(\frac34\right)^{t-2}
    =\frac{7548}{N}.
\]
Theorem~\ref{thm:quantum-fourier} then yields the first bound in Corollary~\ref{cor:lxop-quantum-concrete}.\qed

\subsection{Two-Block Output}
\label{subsec:proof-concrete-quantum-lxop-two}
The case $q=0$ is immediate; assume $q\geq1$. Let $\xi$ denote the density of the complete $\sop[\sigma,2]$ function, whose actual domain has size $D=N/4$. In the one-input calculation of Appendix~\ref{subsec:lxop-one-input-moment}, exactly the $N-1$ construction masks $(\gamma,\gamma)$ with $\gamma\neq0$ contribute. Each has coefficient $2/((N-1)(N-2))$. Thus, as in~\cite[Lemma~8]{Dinur2025:More-SOP:EUROCRYPT},
    \begin{equation}\label{eq:lxop-2-first-moment}
        D\sum_{\alpha\in G^2\setminus\{0\}}
        \left|\widehat{\xi^{(\sigma,2)}_{n,1}}(\alpha)\right|^2
        =\frac{N}{(N-1)(N-2)^2}.
    \end{equation}
    For $2\leq t\leq N/32$, Proposition~19 of the same work, together with input symmetry, gives
    \[
        \sum_{|\alpha|=t}|\widehat\xi(\alpha)|^2
        \leq\binom Dt\,2^{7t/2+3c_t}
            \left(\frac{t}{N-2t}\right)^{3t/2+c_t}
        \eqqcolon A_t,
    \]
    where $c_t=0$ for even $t$ and $c_t=1/2$ for odd $t$. Since $2q\leq N/192$, this estimate applies at every required level.

    For two consecutive levels of the same parity with $4\leq t$ and $t+2\leq2q$, a direct calculation gives
    \begin{equation}\label{eq:lxop-2-moment-interpolation}
    \begin{split}
        \frac{A_{t+2}}{A_t}
        ={}&2^7\frac{(D-t)(D-t-1)}{(t+1)(t+2)}
        \left(\frac{t+2}{N-2t-4}\right)^3\\
        &\qquad\cdot
        \left(\frac{(t+2)(N-2t)}{t(N-2t-4)}\right)^{3t/2+c_t}
        <0.855.
    \end{split}
    \end{equation}
    We give some details for the uniform numerical bound. For fixed $t$, the displayed expression decreases with $N$ when $N\geq192(t+2)$. Substituting $N=192(t+2)$ and using $c_t\leq1/2$ bounds it by
    \[
        \frac{128}{190^3}
        \frac{(47t+96)(47t+95)}{(t+1)(t+2)}
        \left(1+\frac{192}{95t}\right)^{3t/2+1/2}
        <\frac{128\cdot47^2}{190^3}e^{288/95}
        <0.855.
    \]
    The middle inequality holds for $t\geq4$, using $\log(1+x)\leq x-x^2/(2(1+x))$. This proves the bound uniformly in $N$ and $t$.

    We next evaluate the first few levels. For $N\geq2048$, their scaled upper bounds decrease with $N$, and their values at $N=2048$ give
    \[
        NA_2<32.126,\qquad NA_3<1.277,\qquad
        NA_4<5.397,\qquad NA_5<0.406.
    \]
    The right-hand side of~\eqref{eq:lxop-2-first-moment} is smaller than $0.001/N$. Using the ratio $0.855$ in~\eqref{eq:lxop-2-moment-interpolation}, separately for the even and odd tails, yields
    \begin{align*}
        \sum_{1\leq|\alpha|\leq2q}|\widehat\xi(\alpha)|^2
        &\leq\frac{N}{(N-1)(N-2)^2}+\sum_{t=2}^{2q}A_t\\
        &<\frac1N\left(0.001+32.126+1.277
                  +\frac{5.397+0.406}{1-0.855}\right)
        <\frac{74}{N}.
    \end{align*}
    Missing low levels when $q$ is small can be added as positive upper bounds. For $N=1024$, we have $q\leq2$, and direct evaluation of the levels $t\leq4$ gives total second Fourier moment less than $46/N<74/N$. Theorem~\ref{thm:quantum-fourier} now proves the second bound in Corollary~\ref{cor:lxop-quantum-concrete}.\qed

\fi

\end{document}